\documentclass[journal,twoside,web]{ieeecolor}
\usepackage[english]{babel}
\usepackage{lmodern}
\usepackage[utf8]{inputenc}
\usepackage{balance}
\usepackage{xspace}
\usepackage{bbm}
\usepackage{rotating}
\usepackage{url}
\usepackage{csquotes}
\usepackage[noadjust]{cite}
\usepackage{subfig}
\usepackage{paralist}
\usepackage{multirow}
\usepackage{multicol}
\usepackage{amsmath,amssymb,mathtools}
\usepackage{booktabs}
\usepackage{xfrac}
\usepackage{array}
\usepackage{textcomp}
\usepackage{siunitx}
\usepackage{microtype}
\usepackage[usenames,dvipsnames]{xcolor}
\usepackage{safecolours}
\usepackage{pgfplots}
\pgfplotsset{compat=newest,unit code/.code={\si{#1}},plot coordinates/math parser=false,grid style={lightgray}, ylabel right/.style={
        after end axis/.append code={
            \node [rotate=90, anchor=north] at (rel axis cs:1,0.5) {#1};
        }   
    }}
\usepgfplotslibrary{units,external,groupplots,fillbetween}
\usetikzlibrary{positioning,angles,quotes,patterns,shapes,spy, shapes.misc,backgrounds}

\graphicspath{{./figures/}}

\usepackage{scalerel,stackengine}

\makeatletter
\newcommand{\myitem}[1]{%
\item[#1]\protected@edef\@currentlabel{#1}%
}
\makeatother

\newtheorem{theo}{Theorem}
\newtheorem{lem}{Lemma}

\newtheorem{cor}{Corollary}
\newtheorem{remark}{Remark}
\newtheorem{assume}{Assumption}

\DeclareSIUnit{\belmilliwatt}{Bm}
\DeclareSIUnit{\dBm}{\deci\belmilliwatt}

\DeclareMathOperator*{\E}{\mathbb{E}}

\newcommand{\norm}[1]{\lVert#1\rVert}
\newcommand{\abs}[1]{\left\lvert#1\right\rvert}
\newcommand*\diff{\mathop{}\!\mathrm{d}}
\DeclareMathOperator*{\R}{\mathbb{R}}

\newcommand{\transp}{\text{T}}

\DeclareMathOperator*{\argmax}{arg\,max}

\let\originalleft\left
\let\originalright\right
\renewcommand{\left}{\mathopen{}\mathclose\bgroup\originalleft}
\renewcommand{\right}{\aftergroup\egroup\originalright}

\newcommand\figref[1]{Fig.~\ref{#1}}
\newcommand\tabref[1]{Table~\ref{#1}}
\newcommand\secref[1]{Sec.~\ref{#1}}
\newcommand{\eg}{e.g.,\xspace}
\newcommand{\ie}{i.e.,\xspace}

\newcommand{\capt}[1]{\mdseries{\emph{#1}}}

\usepackage{generic}
\usepackage{cite}
\usepackage{amsmath,amssymb,amsfonts}
\usepackage{algorithmic}
\usepackage{graphicx}
\usepackage{algorithm,algorithmic}
\usepackage{hyperref}
\hypersetup{hidelinks=true}
\usepackage{textcomp}
\newcommand{\safeopt}{\textsc{SafeOpt}\xspace}
\newcommand{\colsafe}{\textsc{CoLSafe}\xspace}
\def\BibTeX{{\rm B\kern-.05em{\sc i\kern-.025em b}\kern-.08em
    T\kern-.1667em\lower.7ex\hbox{E}\kern-.125emX}}
\usepackage{fancyhdr}
\newcommand{\mytitle}{\textbf{Accepted final version.}
To appear in \textit{IEEE Transactions on Automatic Control}.\\
\copyright 2025 IEEE. Personal use of this material is permitted. Permission
from IEEE must be obtained for all other uses, in any current or future
media, including reprinting/republishing this material for advertising or
promotional purposes, creating new collective works, for resale or
redistribution to servers or lists, or reuse of any copyrighted component of
this work in other works.}
\begin{document}
\title{Safety and optimality in learning-based control\\ at low computational cost}
\author{Dominik Baumann,~\IEEEmembership{Member,~IEEE,} Krzysztof Kowalczyk, Cristian R. Rojas,~\IEEEmembership{Senior Member,~IEEE,} Koen Tiels, ~\IEEEmembership{Member,~IEEE,} and Pawe\l{} Wachel~\IEEEmembership{Member,~IEEE}
\thanks{This work was supported in part by NWO VIDI 15698, ECSEL 101007311 (IMOCO4.E), the Swedish Research Council under contract number 2023-05170, and the Wallenberg AI, Autonomous Systems and Software Program (WASP) funded by the Knut and Alice Wallenberg Foundation.}
\thanks{Dominik Baumann is with the Cyber-physical Systems Group, Aalto University, Espoo, Finland, and the Department of Information Technology, Uppsala University, Uppsala, Sweden (e-mail: dominik.baumann@aalto.fi).}
\thanks{Krzysztof Kowalczyk and Pawe\l{} Wachel are with the Department of Control Systems and Mechatronics, Wroc\l{}aw University of Science and Technology, Wroc\l{}aw, Poland (e-mails: \{krzysztof.kowalczyk, pawel.wachel\}@pwr.edu.pl).}
\thanks{Cristian R. Rojas is with the Division of Decision and Control Systems, KTH Royal Institute of Technology, Stockholm, Sweden (e-mail: crro@kth.se).}
\thanks{Koen Tiels is with the Department of Mechanical Engineering, Eindhoven University of Technology, Eindhoven, The Netherlands (e-mail: k.tiels@tue.nl).}}

\maketitle
\thispagestyle{fancy}	
\pagestyle{empty}

\begin{abstract}
Applying machine learning methods to physical systems that are supposed to act in the real world requires providing safety guarantees.
However, methods that include such guarantees often come at a high computational cost, making them inapplicable to large datasets and embedded devices with low computational power.
In this paper, we propose \colsafe, a computationally lightweight safe learning algorithm whose computational complexity grows sublinearly with the number of data points.
We derive both safety and optimality guarantees and showcase the effectiveness of our algorithm on a seven-degrees-of-freedom robot arm.
\end{abstract}


\section{Introduction}
\label{sec:introduction}
\IEEEPARstart{W}{hile} learning-based control, \ie learning high-performing control policies from data, has shown remarkable success in recent years~\cite{duan2016benchmarking,jiang2020learning}, theoretical analysis often remains challenging.
However, a theoretical analysis providing safety guarantees is indispensable for future applications like autonomous driving or smart manufacturing, where control systems operate close to and interact with humans.
The area of safe learning~\cite{brunke2022safe,hewing2020learning} addresses this challenge.
Given varying levels of model knowledge, various types of safety certificates can be provided for different algorithms.

The case in which we assume no prior knowledge about the system dynamics but require probabilistic safety guarantees has been most prominently investigated through safe Bayesian optimization algorithms such as \safeopt~\cite{sui2015safe,berkenkamp2016safe,berkenkamp2021bayesian}.
Assuming a parameterized policy and that at least one set of safe parameters is known, this method carefully explores the parameter space while guaranteeing safety with high probability.
A downside of this method is its computational footprint.
First, the algorithm requires us to discretize the parameter space.
An extension toward continuous parameter spaces exists~\cite{duivenvoorden2017constrained}; however, it comes without safety guarantees.
Other extensions aim to make the search process within the discretized parameter space more efficient while still providing guarantees~\cite{kirschner2019adaptive}, at least partially resolving the scalability challenge.
The second challenge in terms of computational complexity is the use of Gaussian process (GP) regression for function estimation.
As GP regression involves a matrix inversion, it scales, in its naive implementation, cubically with the number of data points---in our case, the parameters that have already been explored.
More efficient implementations exist but cannot structurally change the growth in complexity.
This scalability property makes these methods unfit for applications with large datasets.

In our prior work~\cite{baumann2023computationally}, we have provided an alternative to GP regression: the Nadaraya-Watson estimator.
The computational complexity of the Nadaraya-Watson estimator scales linearly with the number of data points and can be reduced to sublinear scaling.
At the same time, we have shown that it is possible to provide similar safety guarantees to those of \safeopt.
In this paper, we extend our prior work from~\cite{baumann2023computationally}.
In particular, we complement the safety with optimality guarantees, which requires significant novel theoretical analysis. 
We further remove the restriction to one-dimensional constraint functions that define safety violations, provide a more extensive simulation study investigating the different exploration behavior of \colsafe and \safeopt, and complement simulation studies with real experiments on a 7-degrees-of-freedom robot arm.
\section{Related work}
For a general introduction to and overview of the safe learning literature, we refer the reader to~\cite{brunke2022safe}.
Here, we focus on safe learning algorithms that leverage Bayesian optimization (BO)~\cite{mockus1978application} due to their high sample efficiency and, thus, better applicability to physical systems~\cite{antonova2017deep,calandra2016bayesian,marco2016automatic}.

Safety can often be formulated as constraints in optimization algorithms: optimize some reward or cost function while, for instance, keeping a minimum distance to obstacles.
Bayesian optimization under unknown constraints~\cite{hernandez2016general,Gelbart2014,gardner2014bayesian,gramacy2011opti,schonlau1998global,picheny2014stepwise,marco2021robot} addresses this problem.
While those approaches incorporate constraints into the optimization, they do accept constraint violations during exploration.
In contrast, we seek to guarantee safety already during exploration.

In response to this, the \safeopt algorithm~\cite{sui2015safe,berkenkamp2021bayesian} provides probabilistic safety guarantees.
In particular, it guarantees that, with high probability, no constraints will be violated.
Since its development, various extensions of \safeopt have been proposed, targeting, for instance, global exploration in case of multiple disjoint safe regions~\cite{baumann2021gosafe,sukhija2023gosafeopt}, effectiveness of the exploration strategy~\cite{wachi2018safe}, or adaptivity~\cite{konig2021applied}.
All these methods leverage GP regression to estimate the reward and constraint function.
While this allows for sample-efficient estimation, it comes with the discussed computational cost due to the matrix inversion.
Thus, scalability is a known problem of \safeopt and BO in general.
Some \safeopt extensions address this limited scalability~\cite{duivenvoorden2017constrained,sui2018stagewise,kirschner2019adaptive}.
However, those extensions still rely on GP regression.

In this work, we address the limited scalability of \safeopt by leveraging the Nadaraya-Watson estimator to estimate the reward and constraint functions.
The Nadaraya-Watson estimator has superior scaling behavior and, thus, effectively improves scalability.
It has been used in prior control-related works, especially in the system identification community~\cite{schuster1979contributions,juditsky1995nonlinear,ljung2006some,mzyk2020wiener}.
We are not aware of any prior work that uses the Nadaraya-Watson estimator for safe learning.
\section{Problem setting}
\label{sec:problem}

We consider a dynamical system
\[
    \diff x(t) = z(x(t),u(t)) \diff t
\]
with unknown system dynamics $z(\cdot)$, state space $x(t)\in\mathcal{X}\subseteq \R^n$ and input $u(t)\in\R^m$.
For this system, we want to learn a policy $\pi$, parameterized by policy parameters $a\in\mathcal{A}\subseteq\R^d$, that determines the control input $u(t)$ given the current system state $x(t)$. The set $\mathcal{A}$ is considered finite.
\begin{remark}
    In control applications, the parameters $a$ may, for instance, correspond to controller parameters, which are, naturally, continuous.
    Thus, in control, the set $\mathcal{A}$ is typically infinite.
    However, we can make it finite through an appropriate discretization, as has been done in many works that leverage Bayesian optimization for learning control policies~\cite{berkenkamp2016safe,berkenkamp2021bayesian,marco2016automatic,calandra2016bayesian,baumann2021gosafe,sukhija2023gosafeopt}.
\end{remark}
The quality of the policy is evaluated by an unknown reward function $f\colon \mathcal{A}\to\R$.
In general, we can solve these kinds of problems through reinforcement learning.
However, this would also mean that, during exploration, we would allow for arbitrary actions that might violate safety constraints.
We define safety in the form of constraint functions $g_i\colon \mathcal{A}\to\R^{m_i}$ whose norm should always be above a safety threshold $c_i$, with $i\in\mathcal{I}_\mathrm{g}=\{1,\ldots,q\}$. 
We can then write the constrained optimization problem as
\begin{equation}
    \label{eqn:constr_max}
    \max_{a\in\mathcal{A}} f(a) \text{ subject to }\norm{g_i(a)}\ge c_i\text{ for all }i\in\mathcal{I}_\mathrm{g},
\end{equation}
where $\norm{\cdot}$ here and in the following denotes the L$_2$-norm.
\begin{remark}
    In most literature on \safeopt and related methods, $g_i$ are defined as functions mapping to the reals, and the safety constraint reduces to $g_i>0$ for all $i\in\mathcal{I}_\mathrm{g}$.
    Our formulation enables us to include constraints like restricting a robot from accessing a region in a 3-dimensional space.
\end{remark}
Solving~\eqref{eqn:constr_max} without knowledge about the system dynamics and reward and constraint functions is generally unfeasible.
Thus, we need to make some assumptions.
Clearly, without any prior knowledge, we have no chance to choose parameters $a$ for a first experiment for which we can be confident that they are safe.
Thus, we assume that we have an initial set of safe parameters.
\begin{assume}
A set $S_0\subset\mathcal{A}$ of safe parameters is known. That is, for all parameters $a$ in $S_0$ we have $\norm{g_i(a)} \ge c_i$ for all $i\in\mathcal{I}_\mathrm{g}$ and $S_0\neq\varnothing$.
\label{ass:safe_seed}
\end{assume}
In robotics, for instance, we often have simulation models available.
As these models cannot perfectly capture the real world, they are insufficient to solve~\eqref{eqn:constr_max}.
However, they may provide us with a safe initial parametrization.
Generally, the set of initial parameters can have arbitrarily bad performance regarding the reward.
Thus, if we imagine a robot manipulator that shall reach a target without colliding with an obstacle, a trivial set of safe parameters would barely move the arm from its initial position.
While this would not solve the task, it would be sufficient as a starting point for our exploration.

Ultimately, we want to estimate $f$ and $g_i$ from data.
Thus, we require some measurements from both after we do experiments.
For notational convenience, let us introduce a selector function
\begin{equation}
    \label{eqn:selector_fcn}
    h(a,i) := \begin{cases}f(a) & \text{ if }i=0\\
    g_i(a) & \text{ if }i\in\mathcal{I}_\mathrm{g},
    \end{cases}
\end{equation}
where the measured quantity is denoted $\hat{h}(a,i) \in \mathbb{R}^{m_i}$, $i \in \mathcal{I}:=\{0\}\cup\mathcal{I}_\mathrm{g}$, and $m_0:=1$.
\begin{assume}
\label{ass:observation_model}
After each experiment with a fixed policy parameterization $a$, we receive measurements $\hat{h}(a,i) = h(a,i) + \omega_i$ for all $i\in\mathcal{I}_\mathrm{g}$, where $\omega_i$, with $i\in\{0,\ldots,q\}$, are zero-mean, conditionally $\sigma^2$-sub-Gaussian random variables of matching dimension.
\end{assume}
That is, we consider $\omega_i$ to be an $m_i$-dimensional real process adapted to the natural filtration $(\mathcal{F}_n)_{n \in \mathbb{N}_0}$ defined by the measurement equation, where $\mathcal{F}_n=\sigma(\hat{h}_{1:n}, a_{1:n})$, and $n\in\mathbb{N}$ is the current experiment.
The $\sigma^2$-sub-Gaussian assumption then implies that for all $n\in\mathbb{N}$ there exists a $\sigma > 0$ such that, for every $m_i$-dimensional real predictable sequence $(\gamma_n)$ (\ie $\gamma_n$ is $\mathcal{F}_{n-1}$-measurable),
\[
    \E\left\{ \exp \left( {\gamma}_n^\top {\omega}_{i,n} \right) \mid \mathcal{F}_{n-1} \right\} \le \exp \left( \frac{{\gamma}_n^\top {\gamma}_n \sigma^2}{2}\right).
\]

Lastly, we require a regularity assumption about the functions $f$ and $g_i$.
\begin{assume}
    The functions $f$ and $g_i$, for all $i\in\mathcal{I}_\mathrm{g}$, are Lipschitz-continuous with known Lipschitz constant $L<\infty$.
    \label{ass:smoothness_assumption}
\end{assume}
Assuming Lipschitz-continuity is relatively common in control and the safe learning literature.
\safeopt can provide its safety guarantees also without knowledge of the Lipschitz constant~\cite{bottero2022information}.
Instead, \safeopt requires a different regularity assumption.
It requires that $f$ and $g_i$ have a known upper bound in a reproducing kernel Hilbert space.
We argue that having an upper bound on the Lipschitz constant is more intuitive.
Generally, the Lipschitz constant can also be approximated from data~\cite{wood1996estimation}.

Note that we clearly could assume individual Lipschitz constants $L_i$.
However, for simplicity, we assume a common Lipschitz constant $L$, which would be the maximum of the individual $L_i$.
\section{Safe and optimal policy learning}
\label{sec:algorithm}

This section introduces \colsafe as an algorithm for safe and optimal policy learning.
We start by introducing the Nadaraya-Watson estimator and analyzing its theoretical properties before embedding the estimator into a safe learning algorithm.
Finally, we state the safety and optimality guarantees.

\subsection{Nadaraya-Watson estimator}
We use the Nadaraya-Watson estimator to estimate the reward function $f$ and the constraint functions $g_i$ given the noisy measurements we receive after doing an experiment.
In particular, the estimate of $h$ at iteration $n$ is given by
\begin{equation}
    \label{eqn:nadaraya-watson}
    \mu_{n}(a',i) := \sum_{t=1}^n\frac{K_\lambda(a', a_t)}{\kappa_n(a')}\hat{h}_t(a,i),
\end{equation}
where
\begin{align*}
    \kappa_n(a') &\coloneqq \sum_{t=1}^nK_\lambda(a',a_t),\\
    K_\lambda(a,a') &\coloneqq \frac{1}{c_K}K\left(\frac{\norm{a-a'}}{\lambda}\right).
\end{align*}
Here, $K$ is the kernel function,~$\lambda$ the bandwidth parameter,~$c_\mathrm{K}$ a constant, and~$\hat{h}_t(a,i)$ the measurements at iteration~$t$.
The kernel needs to meet the following assumption:
\begin{assume}
\label{ass:kernel}
There exist constants $0 < \chi_\mathrm{K} \leq c_\mathrm{K}<\infty$ such that the kernel $K\colon \mathbb{R}_0^+\to\mathbb{R}$ satisfies
%
$\chi_\mathrm{K}\le K(v)\le c_\mathrm{K}$ for all $v\le1$, and $K(v)=0$ for all $v>1$.
\end{assume}
There are numerous classical kernels that satisfy Assumption~\ref{ass:kernel}, such as the box kernel $K_\lambda(v)=\frac{1}{2}$ or the cosine kernel $K_\lambda(v)=\frac{\pi}{4}\cos(\frac{\pi}{2}v)$.

Both for safety and optimality guarantees, we require a bound on how close the Nadaraya-Watson estimate is to the underlying function.
The proofs of all technical results are collected in \secref{sec:proof}.
\begin{lem} 
\label{lem:bounds}
Consider Assumptions~\ref{ass:safe_seed}-\ref{ass:kernel}, where the bandwidth parameter $\lambda$ is a fixed constant. Assume that $|\mathcal{A}|\le D$. Then, for every $0 < \delta <1$, with probability larger than $1 - \delta$, for all $n \in \mathbb{N}$ and $a \in \mathcal{A}$,
\[
\norm{ \mu_n(a,i) - h(a, i) } \leq \beta_n(a,i)
\]
where
\begin{align*}
&\beta_n(a,i) = \\
&\;\; \begin{cases}
+\infty, \qquad \text{if } \kappa_n(a) = 0 \vspace{3pt} \\
\displaystyle L \lambda + \frac{2 \sigma}{\kappa_n(a)} \sqrt{\ln(\delta^{-1} 2^{\frac{m_i}{2}} D)}, \quad \text{if } 0 < \kappa_n(a) \leq 1 \vspace{3pt} \\ 
\displaystyle L \lambda + \frac{2 \sigma}{\kappa_n(a)} \sqrt{\kappa_n(a) \ln \left( \delta^{-1} D \left(1 + \kappa_n(a)\right)^{\frac{m_i
}{2}}\right)},  \text{ else.}
\end{cases}
\end{align*}
\end{lem}

These bounds enable us to provide safety guarantees as we can always make sure to do only experiments with parameters $a$ for which all $g_i$ are above their safety threshold with high probability.
However, to guarantee safety, we further require that the uncertainty after finitely many iterations shrinks to a fixed value below a threshold.
\begin{lem}
\label{lem:bound_uncertainty}
For a given $a \in \mathcal{A}$ and for fixed values of $L$, $\lambda$, $\sigma$, $\delta$, 
after doing $N_{\bar{\beta}}$ experiments with $a_n$ such that $\norm{a-a_n} / \lambda \le 1$, the quantity $\beta_n(a,i)$ in Lemma~\ref{lem:bounds} can be upper bounded by
\begin{equation*}
\bar{\beta} := L \lambda+2\sigma \frac{\max(\alpha_1,\alpha_2(N_{\bar{\beta}}))}
{N_{\bar{\beta}} \chi_\mathrm{K}},
\end{equation*}
where
\begin{align*}
\alpha_1 &\coloneqq \sqrt{\ln \left( \frac{2^{\frac{1}{2}}}{\delta}D \right)}, \\
\alpha_2(N_{\bar{\beta}}) &\coloneqq\sqrt{N_{\bar{\beta}} \chi_\mathrm{K} \ln \left( \frac{(1+N_{\bar{\beta}} \chi_\mathrm{K})^{\frac{1}{2}}}{\delta}D \right)}.
\end{align*}
\end{lem}

Given a desired bound $\bar{\beta}$ ($> L \lambda$) on the uncertainty, we will denote by $N_{\bar{\beta}}$ the smallest positive integer\footnote{This number is well defined because the right-hand side of the inequality tends to $L \lambda$ as $N_{\bar{\beta}} \to \infty$, so the inequality is satisfied for some $N_{\bar{\beta}}$.} such that $\bar{\beta} \ge L \lambda+2\sigma \max(\alpha_1, \alpha_2(N_{\bar{\beta}})) / (N_{\bar{\beta}} \chi_\mathrm{K})$.

\subsection{The algorithm}

We now present the safe learning algorithm.
Starting from the initial safe seed given by Assumption~\ref{ass:safe_seed}, we can do a first experiment and receive a measurement of the reward and constraint functions.
Following~\cite {berkenkamp2021bayesian}, we then leverage the functions' Lipschitz-continuity to generate a set of parameters $a$ that are safe with high probability.
For this, we first construct confidence intervals as
\begin{equation}
    \label{eqn:confidence_interv}
    Q_n(a, i) = \begin{cases}
        [\mu_{n-1}(a, 0)\pm\beta_{n-1}(a,0)]&\text{ if }i=0\\
        [\norm{\mu_{n-1}(a, i)}\pm\beta_{n-1}(a,i)]&\text{ else.}
    \end{cases}
\end{equation}
As the \safeopt algorithm requires that the safe set does not shrink, we further define the contained set as
\begin{equation}
    \label{eqn:contained}
    C_n(a,i)=C_{n-1}(a,i)\cap Q_n(a,i),
\end{equation}
where $C_0(a,i)$ is a subset of $[0,\infty)$ for all $a\in S_0$ and a subset of $\R$ for all $a\in\mathcal{A}\setminus S_0$.
Taking the intersection guarantees that, should the uncertainty increase for later $n$ and, therefore, $Q_n$ increase for some $(a,i)$, the contained set does not increase.
This enables us to define lower and upper bounds as $l_n(a,i)\coloneqq\min C_n(a,i)$ and $u_n(a,i)\coloneqq\max C_n(a,i)$ with which we can update the safe set:
\begin{equation}
    \label{eqn:safe_set_update}
    S_n = \bigcap_{i\in\mathcal{I}_g}\bigcup_{a\in S_{n-1}}\{a'\in\mathcal{A}\mid l_n(a,i)-L\norm{a-a'}\ge c_i\}.
\end{equation}
As we define the safe set in terms of the contained set whose uncertainty cannot increase, the safe set cannot shrink.
By sampling only from this set, we can guarantee that each experiment will be safe with high probability.
However, we also seek to optimize the policy. 
Thus, for a meaningful exploration strategy, we define two additional sets~\cite{berkenkamp2021bayesian}: parameters that are likely to yield a higher reward than our current optimum (potential maximizers, $M_n$) and parameters that are likely to enlarge $S_n$ (potential expanders, $G_n$).
Formally, we define them as
\begin{align}
    \label{eqn:maximizers}
    M_n&\coloneqq \left\{a\in S_n\mid u_n(a,0)\ge \max_{a'\in S_n}l_n(a',0)\right\}\\
    \label{eqn:expanders}
    G_n&\coloneqq \Big\{a\in S_n\mid \exists a' \in\mathcal{A}\setminus S_n, \exists i\in\mathcal{I}_\mathrm{g}\colon \\
    &\qquad\qquad\qquad u_n(a,i)-L\norm{a-a'}\ge c_i \Big\}. \nonumber
\end{align}
Given those sets, we select our next sample location as
%
\begin{equation}
    \label{eqn:next_sample}
    a_n = 
        \argmax\limits_{a'\in M_n\cup G_n}\max_{i\in\mathcal{I}}w_n(a',i), 
\end{equation}
with 
$w_n(a,i) = u_n(a,i) - l_n(a,i)$, which is basically a variant of the upper confidence bound algorithm~\cite{srinivas2012gaussian}.
Further, we can, at any iteration, obtain an estimate of the optimum through
\begin{equation}
    \label{eqn:optimum}
    \hat{a}_n = \argmax_{a\in S_n}l_n(a,0).
\end{equation}
The entire algorithm is summarized in Algorithm~\ref{alg:diffopt}.

\begin{algorithm}
\small 
\caption{Pseudocode of \colsafe.}
\label{alg:diffopt}
\begin{algorithmic}[1]
\STATE \textbf{Input:} Domain $\mathcal{A}$, Safe seed $S_0$, Lipschitz constant $L$, optimality bound $\bar{\beta}$
\WHILE{True} 
\STATE Update safe set with~\eqref{eqn:safe_set_update}
\STATE Update set of potential maximizers with~\eqref{eqn:maximizers}
\STATE Update set of potential expanders with~\eqref{eqn:expanders}
\STATE Select $a_n$ with~\eqref{eqn:next_sample}
\STATE Receive measurements $\hat{f}(a_n)$, $\hat{g}_i(a_n)$, for all $i\in\mathcal{I}_\mathrm{g}$
\STATE Update Nadaraya-Watson estimator~\eqref{eqn:nadaraya-watson} with new data
\IF{$S_{n-1} = S_n$ and $w_n(a,i) \le 2\bar{\beta}$ for all $a\in M_n \cup G_n, i\in\mathcal{I}$}
\RETURN Best guess~\eqref{eqn:optimum}
\ENDIF
\ENDWHILE
\end{algorithmic}
\end{algorithm}

Before starting our main result, given a desired optimality bound $\bar{\beta} > L \lambda$, we need to define the reachable set $R_{\bar{\beta}}(S)$, which contains the parameters $a\in\mathcal{A}$ that we could directly safely explore from $S$ if we knew $g(a,i)$ with $2\bar{\beta}$-precision:
\begin{align}
    \label{eqn:reachability}
    &R_{\bar{\beta}}(S) := \\
    &S \cup \bigcap_{i\in\mathcal{I}_\mathrm{g}}\left\{a\in\mathcal{A}\mid \exists a'\in S\colon g(a',i) - 2\bar{\beta} - L\norm{a-a'}\ge c_i\right\}.\nonumber
\end{align}
We further define $R_{\bar{\beta}}^n(S):=R_{\bar{\beta}}^{n-1}(R_{\bar{\beta}}(S))$, $R^1_{\bar{\beta}}(S):=R_{\bar{\beta}}(S)$, and the closure of $R_{\bar{\beta}}$, $\bar{R}_{\bar{\beta}}(S)\coloneqq \lim_{n\to\infty}R_{\bar{\beta}}^n(S)$, which is the largest set that can be safely explored starting from $S$.

We can then prove the following properties.

\begin{theo}
\label{thm:safety}
Under Assumptions~\ref{ass:safe_seed}--\ref{ass:kernel}, when following Algorithm~\ref{alg:diffopt}, and given a desired optimality bound $\bar{\beta} > L \lambda$, the following holds with probability at least $1-\delta$:
\begin{enumerate}
    \item \emph{Safety:} For all $n\ge 0$ and all $i\in\mathcal{I}_\mathrm{g}$: $\norm{g_i(a_n)}\ge c_i$.
    \item \emph{Optimality:} Algorithm~\ref{alg:diffopt} stops in finite time, and returns $\hat{a}_n$ for which $f(\hat{a}_n)\ge f(a^*)-2\bar{\beta}$,
\end{enumerate}
where $a^*$ is the parameterization that maximizes $f$ in the safely reachable set $\bar{R}_{\bar{\beta}}(S_0)$.
\end{theo}

In the proof, we also give an upper bound for the maximum number of iterations after which \colsafe stops and satisfies the optimality guarantees with high probability.

\begin{remark}
    The parameters $\bar{\beta}$ and $\lambda$ can be chosen by the user.
    Thus, we can reach the safe optimum with arbitrary precision---at the expense of letting the algorithm run for more time steps.
\end{remark}
\section{Evaluation}

In the evaluation, we seek to demonstrate that \colsafe \emph{(i)} is generally capable of learning control policies for dynamical systems, \emph{(ii)} is computationally more efficient than \safeopt, and \emph{(iii)} is applicable to real experiments.
Further, we want to compare the exploration strategies of \colsafe and \safeopt.

\subsection{Setup}
We evaluate \colsafe and compare it to \safeopt in both simulation and real experiments.
For both, we consider the same setup that has also been used in an earlier work on \safeopt~\cite{sukhija2023gosafeopt}.
In particular, both algorithms are supposed to learn a control policy for a seven-degrees-of-freedom Franka robot arm.
The goal is to let the arm of the robot reach a desired set point without colliding with an obstacle.
For both simulations and real experiments, we consider an operational space impedance controller~\cite{siciliano2008springer} with impedance gain $K$.

For \safeopt, we adopt the parameter settings from~\cite{sukhija2023gosafeopt}, \ie we use a Mat\'{e}rn kernel with $\nu=1.5$.
For \colsafe, we use the same kernel with a length scale of 0.05, set the bandwidth parameter $\lambda=0.5$, and the Lipschitz constant $L=1.75$.
To meet Assumption~\ref{ass:kernel}, we only use the output of the kernel for the Nadaraya-Watson estimator if $\norm{a-a'}<1$; else, we set it to 0.

\subsection{Simulations}

In simulation, we obtain the entries of the impedance gain vector $K$ through a linear quadratic regulator (LQR) approach.
In particular, following~\cite{sukhija2023gosafeopt}, we define the $Q$ and $R$ matrices as
\begin{align*}
    Q &= \begin{pmatrix}
        Q_\mathrm{r} & 0\\ 0 & \kappa_\mathrm{d}Q_\mathrm{r}
    \end{pmatrix},\text{ where } Q_\mathrm{r} = 10^{q_\mathrm{c}}I_3\\
    R &= 10^{r-2}I_3.
\end{align*}
Here, we heuristically set $\kappa_\mathrm{d}=0.1$ and are then left with $q_\mathrm{c}\in \R$ and $r\in \R$ as the tuning parameters.
Hence, the parameter space is $d=2$ for this experiment.
For the system model, we assume that a feedback linearization approach~\cite{siciliano2008springer} is used.
Hence, the assumed system matrices are
\[
    A = \begin{pmatrix}
        0 & I_3\\ 0 & 0
    \end{pmatrix}\quad B = \begin{pmatrix}
        0\\ I_3
    \end{pmatrix}.
\]
As we actually use an impedance controller instead of feedback linearization, this derived controller will be suboptimal, and there is a need for further tuning.
However, it can provide us with a safe seed.

We run both \colsafe and \safeopt in the simulation environment and show the evolution of the safe set in \figref{fig:exploration_comparison}.
Already after the first iteration, \colsafe has a relatively large safe set in comparison to \safeopt.
As we continue the exploration, the safe set grows for both algorithms.
However, comparing the safe set at iterations~60 and~200, we see that \colsafe does not expand the safe set further, \ie it already converged.
Thus, we can conclude that \colsafe is significantly faster in exploring the safe set while providing comparable safety guarantees.
However, this result depends on the chosen kernels and hyperparameters and is difficult to generalize.

\begin{figure*}
    \centering
    \subfloat[\colsafe after 1 iteration.]{\includegraphics[width=0.3\textwidth]{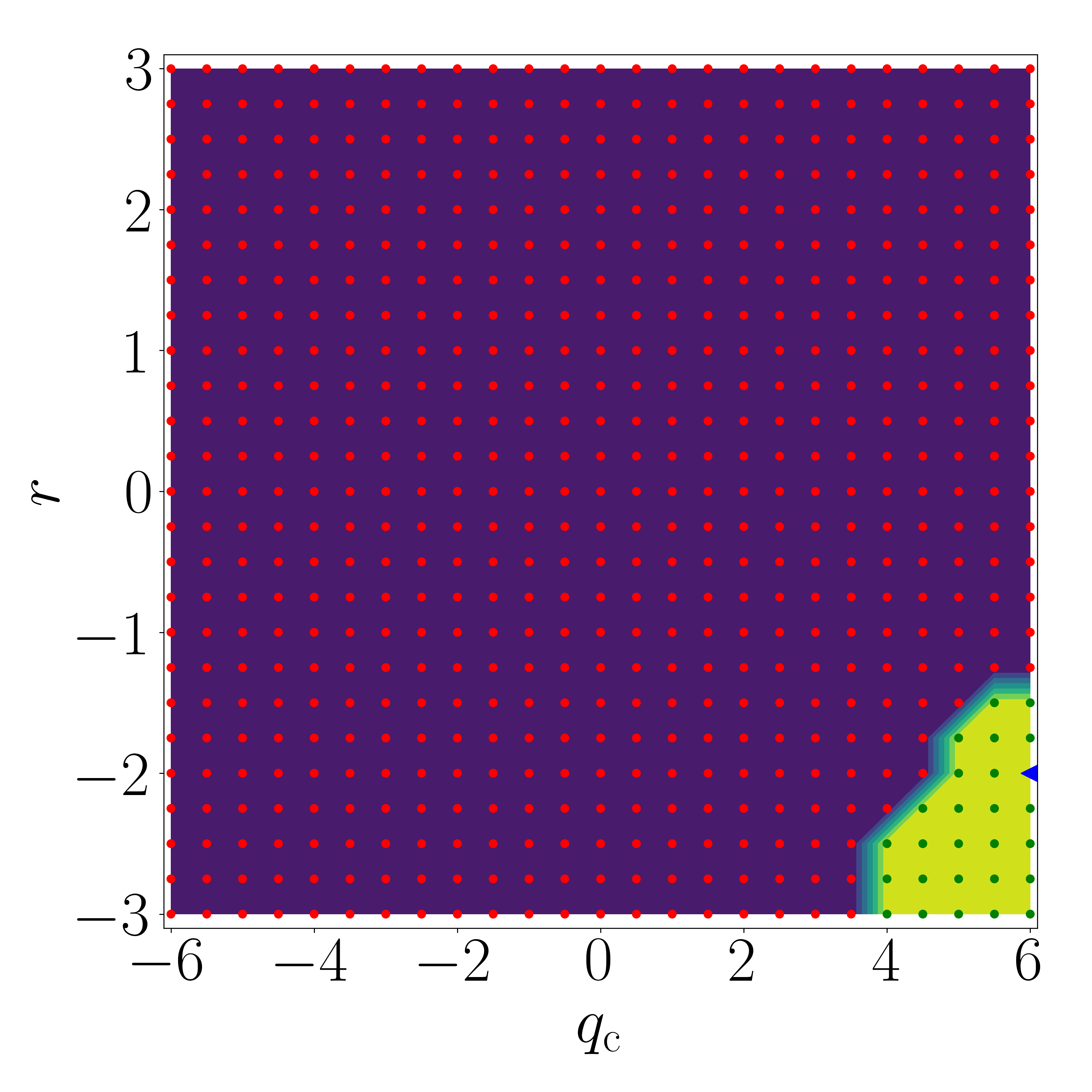}}
    \subfloat[\colsafe after 60 iterations.]{
    \includegraphics[width=0.3\textwidth]{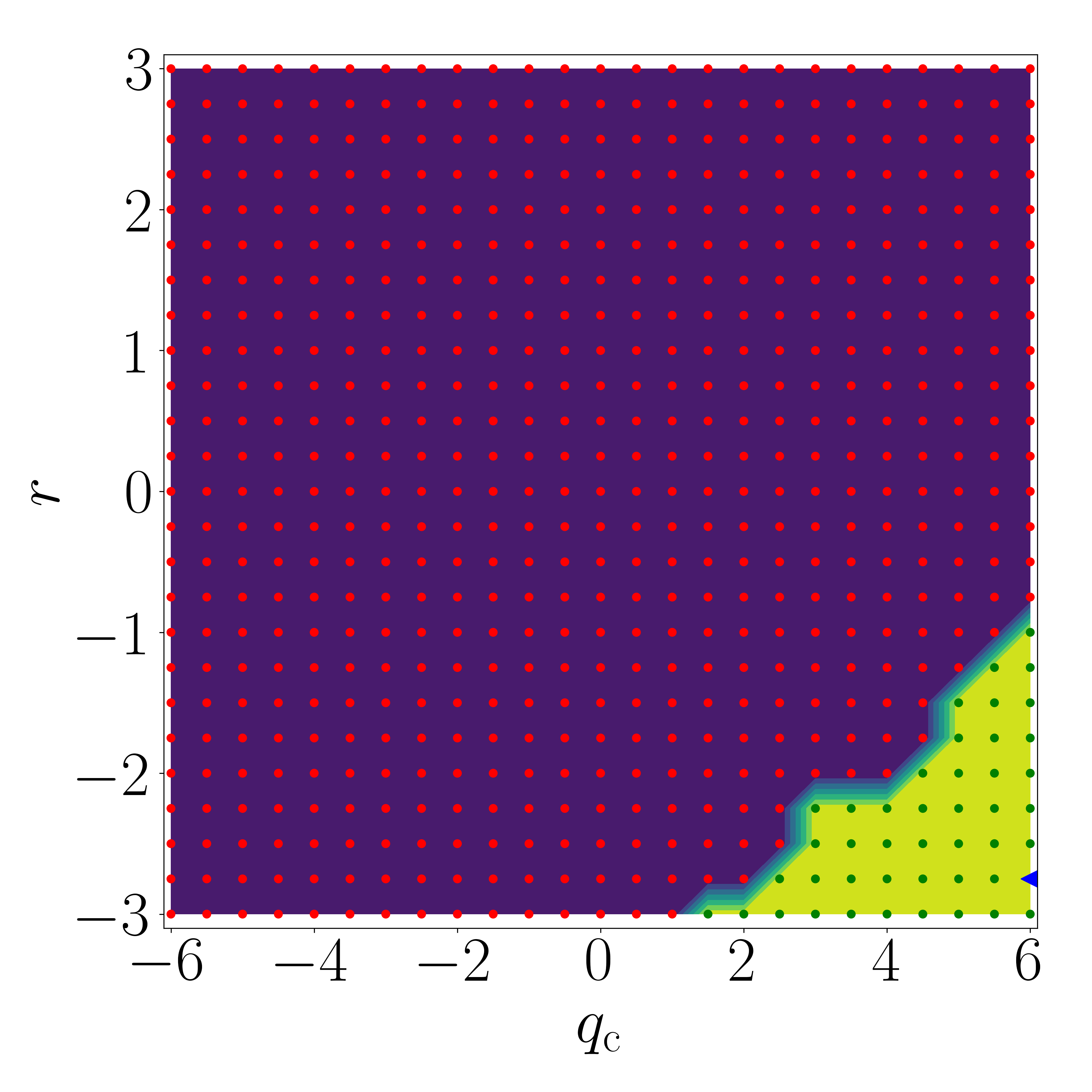}}
    \subfloat[\colsafe after 200 iterations.]{
    \includegraphics[width=0.3\textwidth]{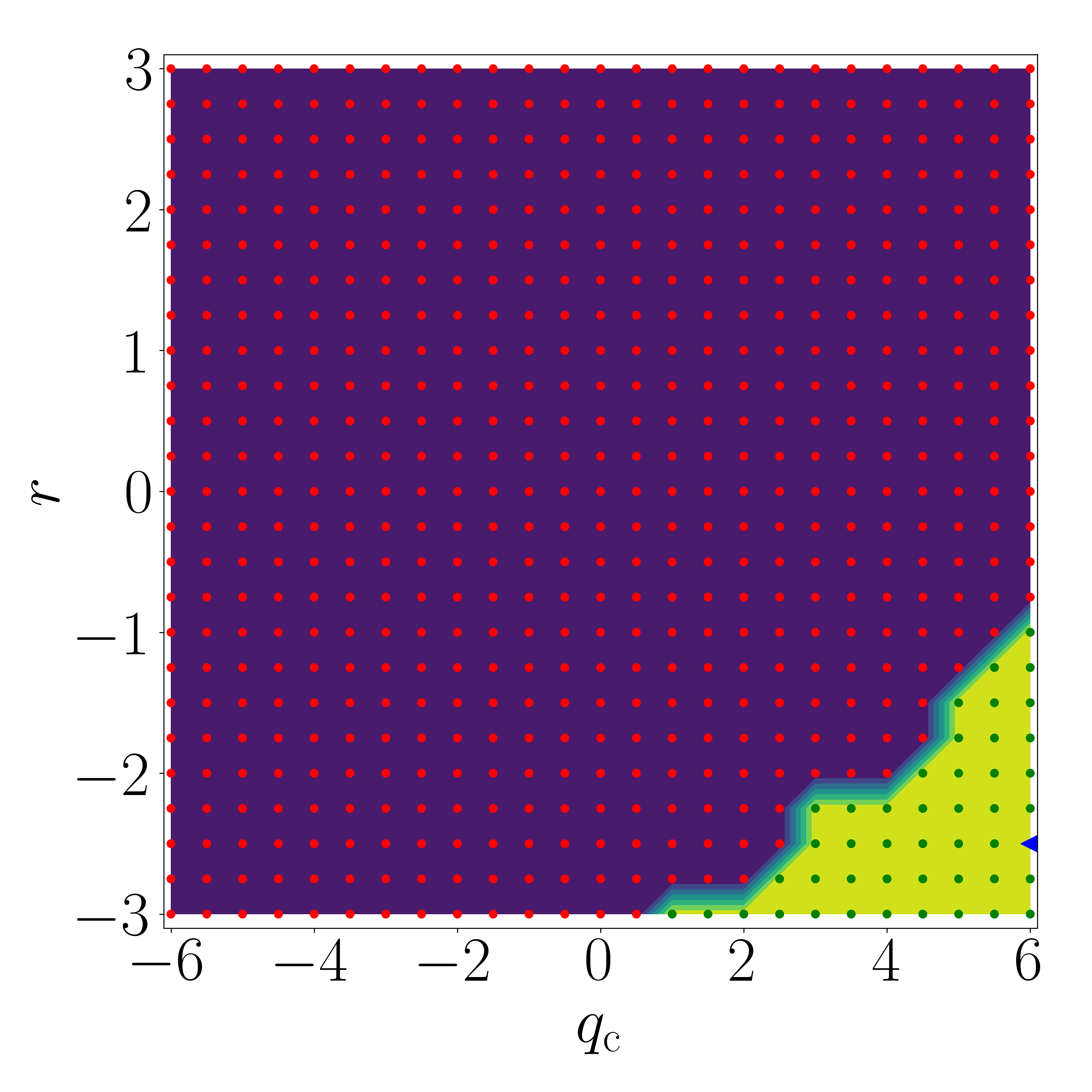}}\\
    \subfloat[\safeopt after 1 iteration.]{
    \includegraphics[width=0.3\textwidth]{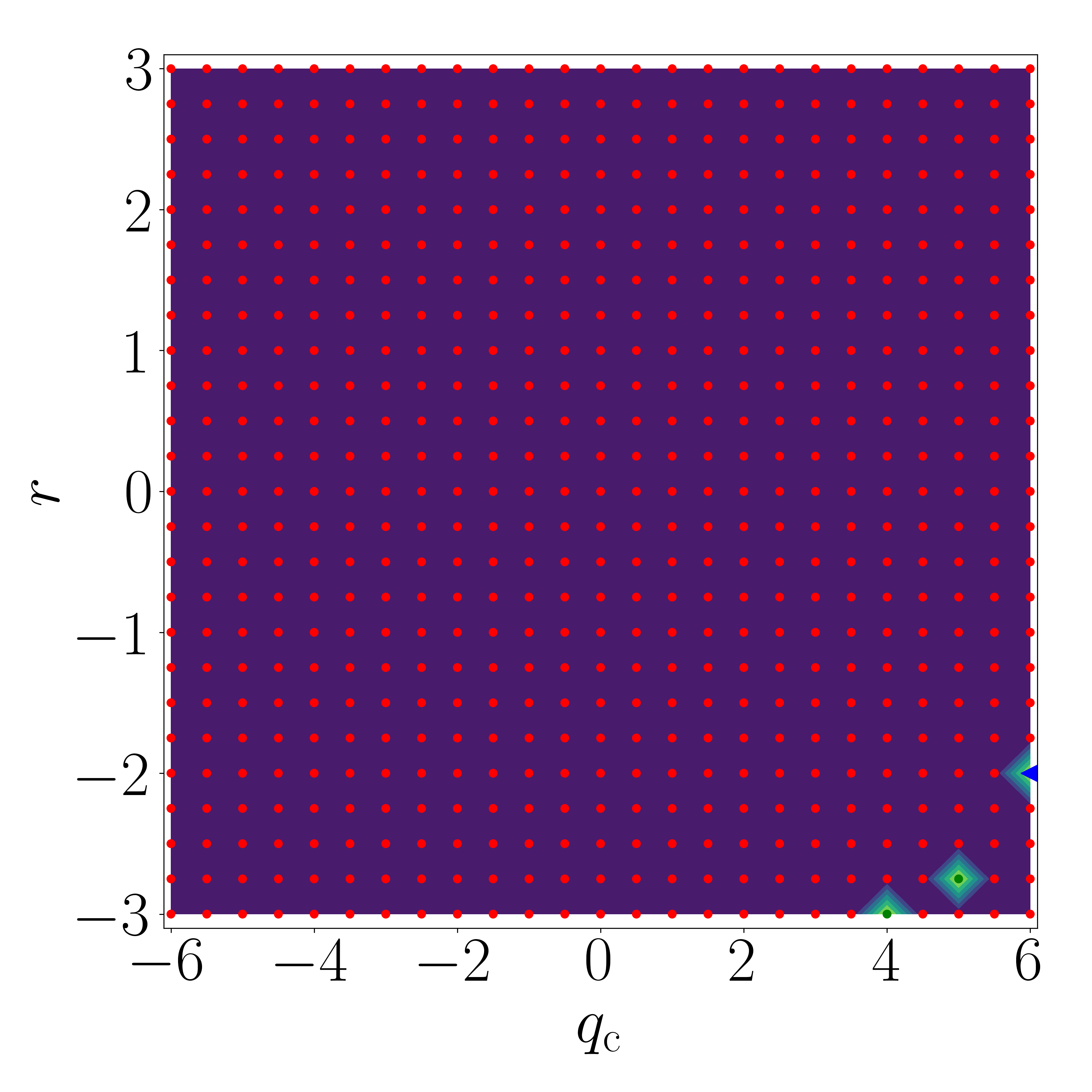}}
    \subfloat[\safeopt after 60 iterations.]{
    \includegraphics[width=0.3\textwidth]{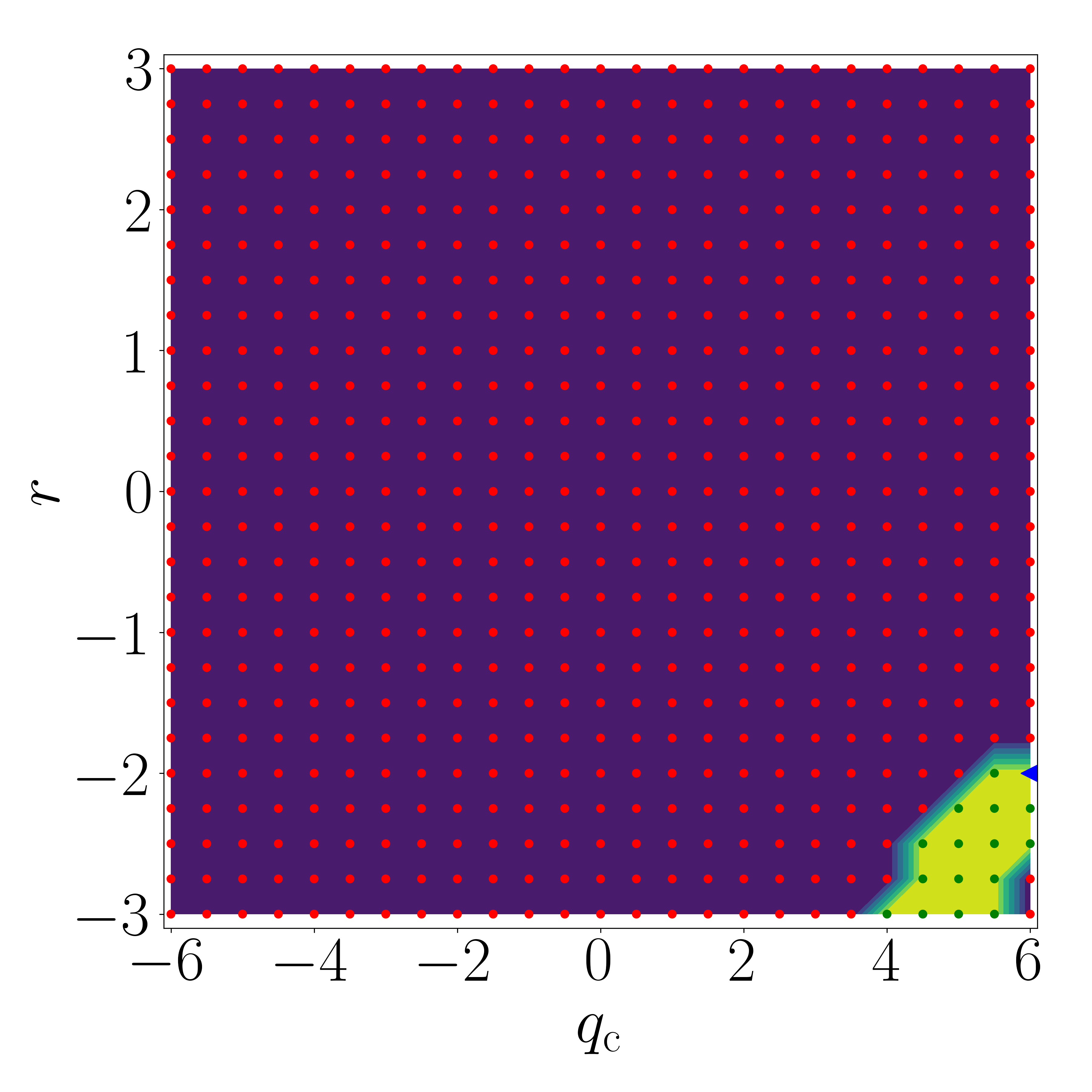}}
    \subfloat[\safeopt after 200 iterations.]{
    \includegraphics[width=0.3\textwidth]{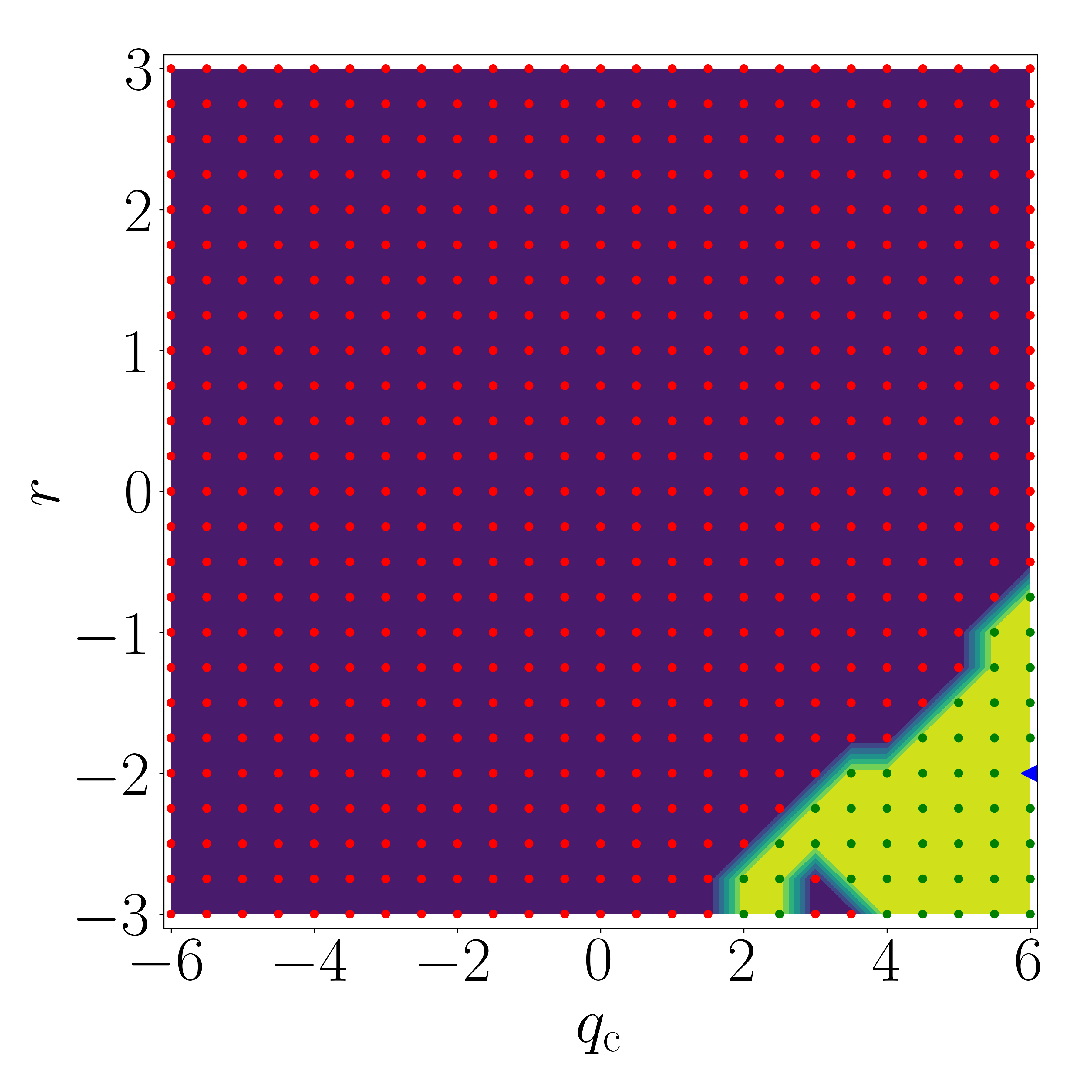}}
    \caption{Exploration behavior comparison of \colsafe and \safeopt. \capt{We can see that \colsafe explores the safe set significantly faster. The yellow regions mark the safe set, and the blue triangle is the current optimum.}}
    \label{fig:exploration_comparison}
\end{figure*}

We let both algorithms run for a total of 410 episodes.
In \figref{fig:eval_safe_set}, we show the safe set of both algorithms after 400 episodes.
While \colsafe had already converged after 60 iterations, \safeopt continued to increase the safe set beyond the safe set that \colsafe can explore.
Thus, we conclude that \colsafe is more conservative.
However, as we discussed, \safeopt requires significantly more samples to achieve the same result.
Further, \safeopt requires more time for the individual optimization steps.
We show this in \figref{fig:time_comparison}.
We measure the time needed to update the safe set and suggest the next policy parameters for both algorithms.
Especially during later iterations, when the data set is larger, \safeopt requires significantly more time for each update step.
In particular, in the last iteration, it takes \safeopt around \SI{1}{\hour} to suggest the next sample point, while \colsafe needs only \SI{12.3}{\second}. 
The times were measured on a standard laptop.

\begin{figure}
    \centering
    \subfloat[\colsafe after 400 episodes.\label{sfig:our_400}]{\includegraphics[width=0.24\textwidth]{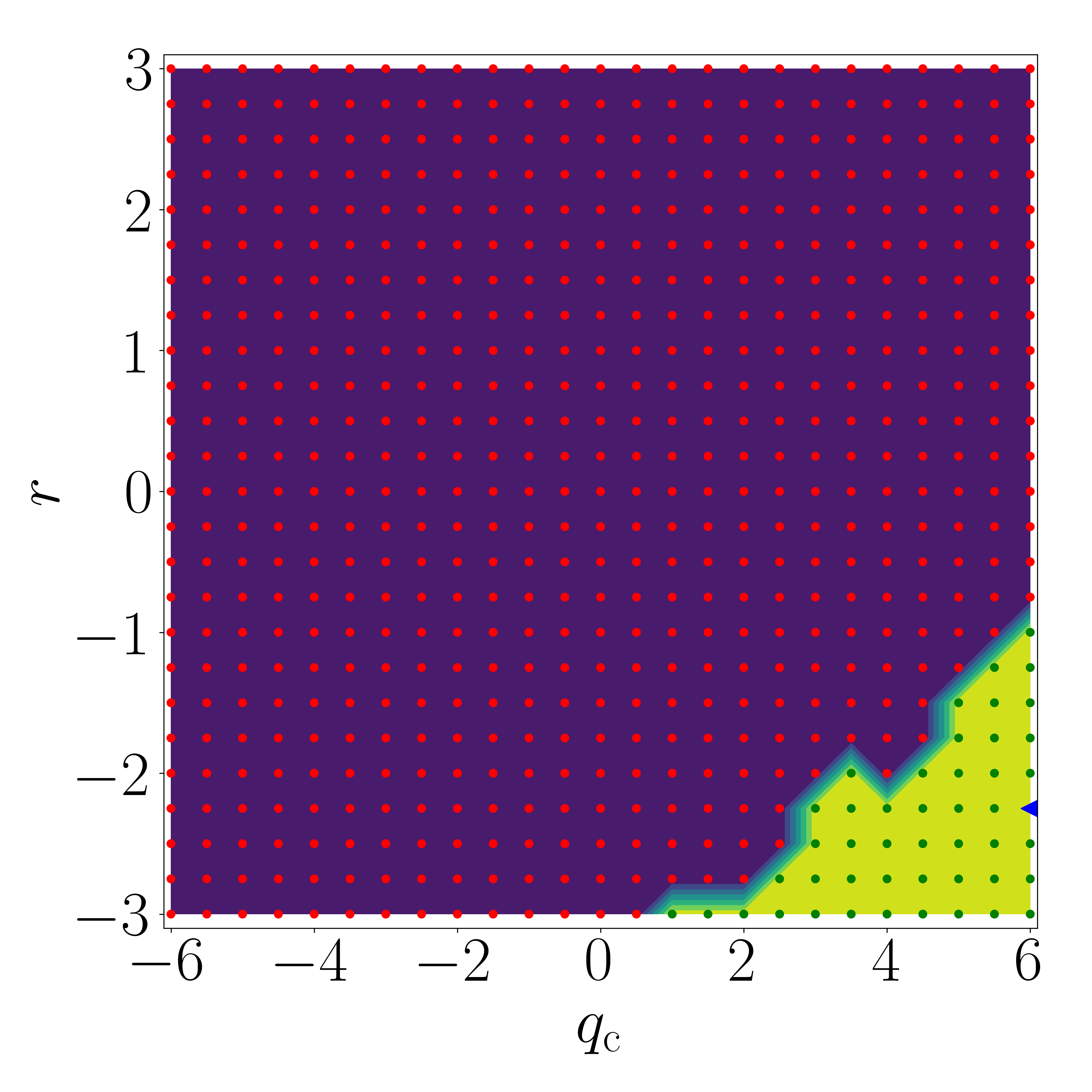}}
    \subfloat[\safeopt after 400 episodes.\label{sfig:safeopt_400}]{\includegraphics[width=0.24\textwidth]{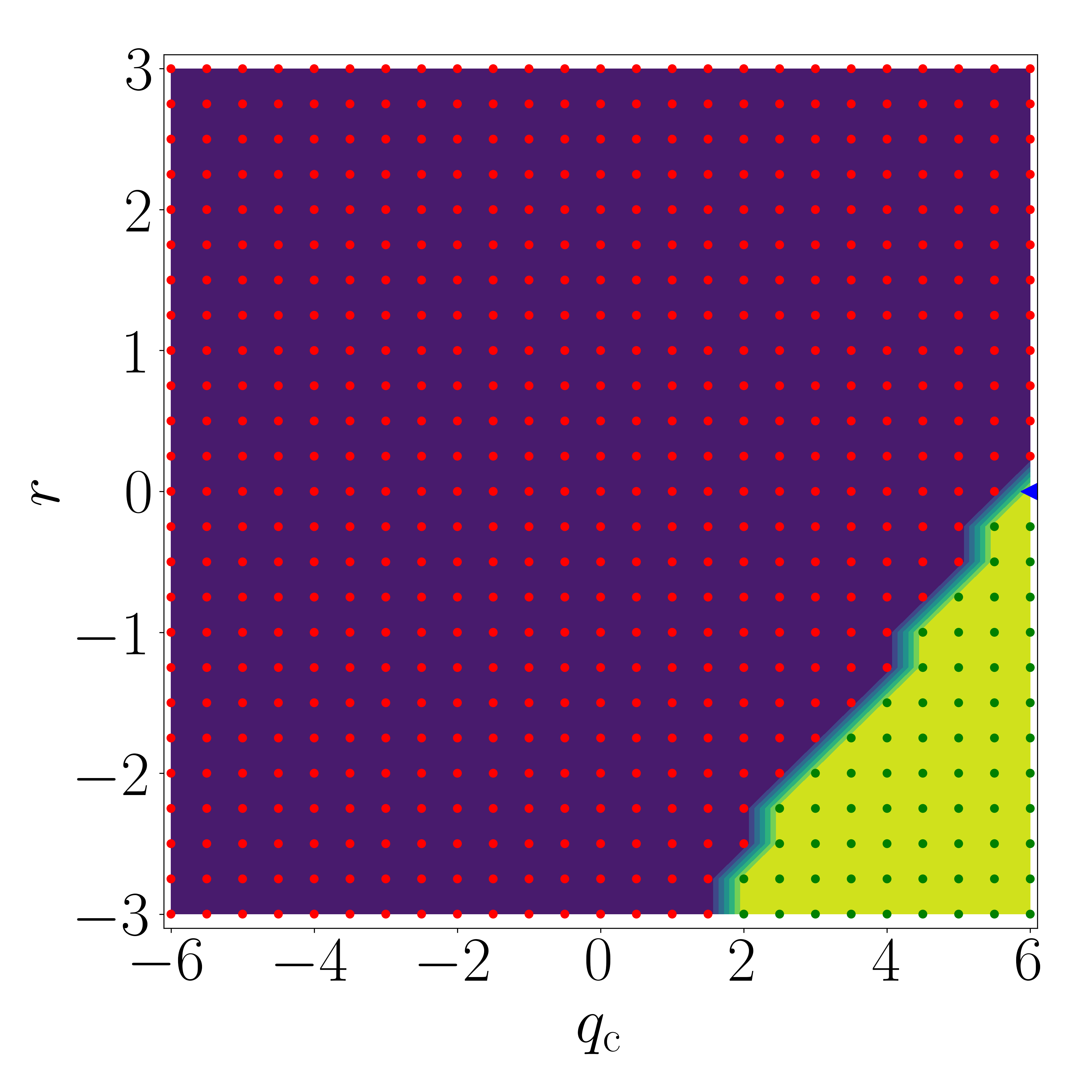}}
    \caption{Performance of \colsafe and \safeopt on the simulated Franka robot. \capt{Both algorithms can similarly explore the safe region, while the computational footprint of \colsafe is significantly lower.}}
    \label{fig:eval_safe_set}
\end{figure}

\begin{figure}
    \centering
    \input{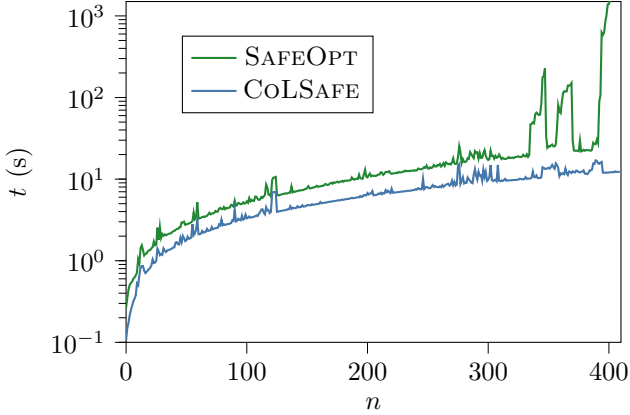}
    \caption{Time complexity of \safeopt and \colsafe. \capt{We measure the time $t$ for updating the safe set and suggesting the next candidate point in each iteration $n$. For higher iterations, \ie more data, \safeopt requires significantly more time. Note the logarithmic scaling of the $y$-axis.}}
    \label{fig:time_comparison}
\end{figure}

These findings align with the general scaling behavior of GPs and the Nadaraya-Watson estimator.
As discussed before, GPs scale cubically with the number of data points. 
Hence, in our case, we have $\mathcal{O}(|\mathcal{A}_\mathrm{expl}|^3)$, where $|\mathcal{A}_\mathrm{expl}|$ is the cardinality of the set of already explored parameters.
Through efficient implementations, this complexity can be reduced to, e.g., quadratic scaling $\mathcal{O}(|\mathcal{A}_\mathrm{expl}|^2)$~\cite{davies2015effective}.
However, we cannot avoid an exponent larger than one unless we leverage approximation techniques~\cite{liu2020gaussian}---however, then we also lose the theoretical guarantees.
The Nadaraya-Watson estimator, on the other hand, already in its naive implementation, exhibits linear scaling properties, i.e., the computational complexity grows with $\mathcal{O}(|\mathcal{A}_\mathrm{expl}|)$~\cite{rao1996pac}.
Through efficient implementations, this can even be reduced to $\mathcal{O}(\log(|\mathcal{A}_\mathrm{expl}|)$~\cite{rao1996pac}.

This evaluation suggests an interesting trade-off.
While the computations for \colsafe are cheaper, the bounds are more conservative, and thus, only a smaller safe set can be explored.
That is, when computing resources are not a concern, and the dimensionality and expected number of data samples are relatively low, but we need to explore the largest possible safe set to increase the chance of finding the global optimum, \safeopt is the better choice.
However, for higher dimensional systems and restricted computing resources, \eg when learning policies on embedded devices, \colsafe is superior.
Further, the reduced computational complexity allows for more exploration.
Thus, \colsafe is less dependent on its initialization.

Lastly, we provide an ablation study to evaluate the influence of the choice of kernel and bandwidth parameter $\lambda$ in \figref{fig:ablation}.
In particular, we test the Mat\'{e}rn kernel and the box kernel for bandwidth parameters $\{1, 1.5, 1.75\}$.
It can be seen that both the choice of the kernel and the bandwidth parameter significantly influence the quality of the solution that can be obtained. 
In particular, a smaller bandwidth parameter results in faster improvements in the objective function.
Moreover, in this example, the very simple box kernel performs better than the more sophisticated and, for GPs, very popular Mat\'{e}rn kernel.

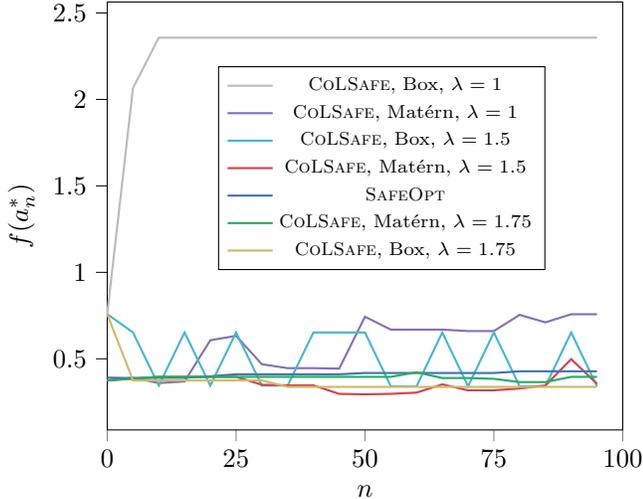
\begin{figure}
    \centering
\begin{tikzpicture}

\begin{axis}[
tick align=outside,
tick pos=left,
x grid style={darkgray176},
xmin=0, xmax=20,
xlabel=$n$,
ylabel=$f(a_n^*)$,
ytick style={color=black},
xticklabel={\pgfmathparse{\tick*5}\pgfmathprintnumber{\pgfmathresult}},
legend style={at={(0.85,0.85)}}
]
\addplot [thick, solin-colour7]
table {%
0 0.758575078840016
1 2.0652284810864
2 2.35714285714286
3 2.35714285714286
4 2.35714285714286
5 2.35714285714286
6 2.35714285714286
7 2.35714285714286
8 2.35714285714286
9 2.35714285714286
10 2.35714285714286
11 2.35714285714286
12 2.35714285714286
13 2.35714285714286
14 2.35714285714286
15 2.35714285714286
16 2.35714285714286
17 2.35714285714286
18 2.35714285714286
19 2.35714285714286
};
\addlegendentry{\scriptsize{\colsafe, Box, $\lambda=1$}}
\addplot [thick, solin-colour4]
table {%
0 0.374663532656303
1 0.390841695684511
2 0.359437123621313
3 0.368856762503266
4 0.607088584990098
5 0.632664240920088
6 0.468398804640578
7 0.445802566238741
8 0.445802566238741
9 0.443152666301776
10 0.743380603709874
11 0.668781465302424
12 0.668781465302424
13 0.668781465302424
14 0.660352363178982
15 0.660352363178982
16 0.754043471911284
17 0.710580064602428
18 0.757612840283957
19 0.757612840283957
};
\addlegendentry{\scriptsize{\colsafe, Mat\'{e}rn, $\lambda=1$}}
\addplot [thick, solin-colour6]
table {%
0 0.758575078840016
1 0.652225548954768
2 0.345039816704852
3 0.652225548954768
4 0.345039816704852
5 0.652225548954768
6 0.345039816704852
7 0.345039816704852
8 0.652225548954768
9 0.652225548954768
10 0.652225548954768
11 0.340857249850956
12 0.340857249850956
13 0.652225548954768
14 0.340857249850956
15 0.652225548954768
16 0.340857249850956
17 0.340857249850956
18 0.652225548954768
19 0.340857249850956
};
\addlegendentry{\scriptsize{\colsafe, Box, $\lambda=1.5$}}
\addplot [thick, solin-colour3]
table {%
0 0.374663532656303
1 0.385084171386038
2 0.395523390017352
3 0.395523390017352
4 0.395523390017352
5 0.395523390017352
6 0.349832601326583
7 0.346267497423473
8 0.346267497423473
9 0.297068056456411
10 0.2947371359549
11 0.297341620946722
12 0.304377085970945
13 0.351981383843497
14 0.317610490956351
15 0.317610490956351
16 0.328195109486926
17 0.345896178752956
18 0.498099451238615
19 0.355838636493774
};
\addlegendentry{\scriptsize{\colsafe, Mat\'{e}rn, $\lambda=1.5$}}
\addplot [thick, solin-colour1]
table {%
0 0.390841695684511
1 0.38815572228805
2 0.390249838565598
3 0.390249838565598
4 0.398838201765247
5 0.41041793484385
6 0.41041793484385
7 0.41041793484385
8 0.41041793484385
9 0.41041793484385
10 0.417821342027618
11 0.417821342027618
12 0.417821342027618
13 0.417821342027618
14 0.417821342027618
15 0.417821342027618
16 0.427675217926598
17 0.427675217926598
18 0.427675217926598
19 0.427675217926598
};
\addlegendentry{\scriptsize{\safeopt}};
\addplot [thick, solin-colour2]
table {%
0 0.374663532656303
1 0.385694099121704
2 0.395523390017352
3 0.395523390017352
4 0.395523390017352
5 0.395523390017352
6 0.395523390017352
7 0.395523390017352
8 0.395523390017352
9 0.395523390017352
10 0.395523390017352
11 0.395523390017352
12 0.422205224603807
13 0.388726073494861
14 0.388726073494861
15 0.384118426099576
16 0.365142636731432
17 0.365142636731432
18 0.395523390017352
19 0.395523390017352
};
\addlegendentry{\scriptsize{\colsafe, Mat\'{e}rn, $\lambda=1.75$}}
\addplot [thick, solin-colour5]
table {%
0 0.758575078840016
1 0.374989181222972
2 0.374989181222972
3 0.374989181222972
4 0.374989181222972
5 0.374989181222972
6 0.374989181222972
7 0.337504119394951
8 0.337504119394951
9 0.337504119394951
10 0.337504119394951
11 0.337504119394951
12 0.337504119394951
13 0.337504119394951
14 0.337504119394951
15 0.337504119394951
16 0.337504119394951
17 0.337504119394951
18 0.337504119394951
19 0.337504119394951
};
\addlegendentry{\scriptsize{\colsafe, Box, $\lambda=1.75$}}
\end{axis}

\end{tikzpicture}
    \caption{Ablation study. \capt{Believed optima over number of iterations for different kernels and bandwidth parameters $\lambda$.}}
    \label{fig:ablation}
\end{figure}

\subsection{Hardware experiments}

For the hardware experiments, we consider an impedance gain
\[
    K = \text{diag}(K_x, K_y, K_z, 2\sqrt{K_x}, 2\sqrt{K_y}, 2\sqrt{K_z}),
\]
where $K_x=\alpha_xK_{rx}$, $K_{rx}$ is a reference value taken from Franka's impedance controller, and $\alpha_x\in[0, 1.2]$ is our tuning parameter.
Parameters $K_y$ and $K_z$ are modeled in the same way.
That is, for $\alpha_x=\alpha_y=\alpha_z=1$, we recover the controller provided by the manufacturer, which we take as our initial safe estimate.
We then aim to improve upon this initial controller.
The setup results in a three-dimensional parameter space.

We leverage the reward and constraint functions from~\cite{sukhija2023gosafeopt} that encourage accurate tracking of the desired path while avoiding hitting an obstacle.
Given these reward and constraint functions, we run both \colsafe and \safeopt.
Always after five experiments, we do an additional experiment with the controller parameters that the algorithms currently believe to be optimal.
We report the rewards of these additional experiments in \figref{fig:real_experiments_reward}.
First of all, we notice that \colsafe can successfully optimize the reward function also in hardware experiments.
Further, similar to the simulations, also in experiments, \colsafe does not incur any failures.
We further see that \colsafe can improve the reward faster than \safeopt, which is in line with our earlier observations that it is faster in growing the initial safe set, where we then also have a higher chance of finding a new optimum.

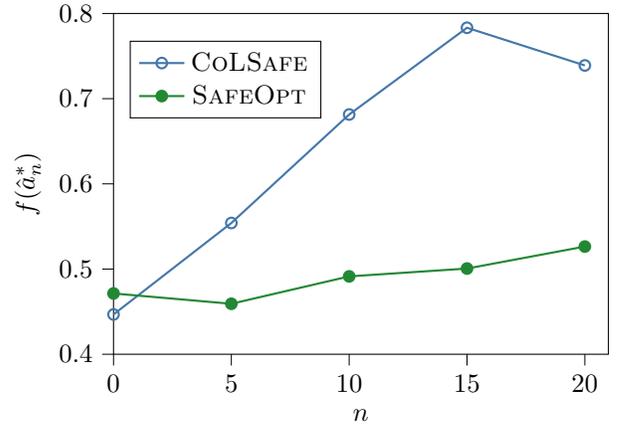
\begin{figure}
    \centering
\begin{tikzpicture}

\definecolor{darkgray176}{RGB}{176,176,176}
\definecolor{darkorange25512714}{RGB}{255,127,14}
\definecolor{lightgray}{RGB}{211,211,211}
\definecolor{steelblue31119180}{RGB}{31,119,180}

\begin{axis}[
tick align=outside,
tick pos=left,
x grid style={darkgray176},
xmin=0, xmax=21,
xtick style={color=black},
y grid style={darkgray176},
height=0.25\textheight,
width = 0.45\textwidth,
ymin=0.4, ymax=0.8,
ytick style={color=black},
ylabel=$f(\hat{a}^*_n)$,
xlabel=$n$,
legend style={at={(0.42,0.92)}}
]
\addplot [thick, mark=o, tol-colour1]
table {%
0 0.446679061851047
5 0.55406706163356
10 0.681405925786021
15 0.783241097551212
20 0.73899187331973
};
\addlegendentry{\colsafe};

\addplot [thick, mark=*, tol-colour3]
table {%
0 0.471352171638027
5 0.459231098601949
10 0.49138543324398
15 0.500592233988356
20 0.526401371950271
};
\addlegendentry{\safeopt}
\end{axis}

\end{tikzpicture}
    \caption{\colsafe vs.\ \safeopt in real experiments. \capt{We show the reward achieved by the currently believed optimal parameters after varying number of iterations. We can see that \colsafe can increase the reward faster than \safeopt.}}
    \label{fig:real_experiments_reward}
\end{figure}

\section{Conclusion}

We propose a novel algorithm for safe exploration in reinforcement learning.
Under assumptions that are comparable to other safe learning methods, such as the popular \safeopt algorithm, we can provide similar high-probability safety and optimality guarantees.
However, our algorithm has significantly lower computational complexity, simplifying its use on embedded devices and for high-dimensional systems.
Moreover, the algorithm can explore the safe set faster than \safeopt.
On the other hand, the final safe set that \safeopt explores is larger than that of \colsafe.
Thus, we trade the lower computational complexity and the faster initial search for a more conservative final result.
In practice, it depends on the application and the available computational resources, which algorithm is preferable.

Herein, we basically adopted the \safeopt algorithm and replaced the Gaussian process estimates with the Nadaraya-Watson estimator. 
For future work, it might be interesting to combine those estimates also with more recent advances of \safeopt such as~\cite{konig2021applied,duivenvoorden2017constrained,sui2018stagewise,sukhija2023gosafeopt}.

Further, we here considered a constant bandwidth parameter $\lambda$.
This parameter could also be tuned to make the algorithm less conservative, and that way, potentially come closer to the performance of \safeopt.

Lastly, we have assumed knowledge of the Lipschitz constant or an upper bound on the Lipschitz constant.
In future work, we seek to estimate the Lipschitz constant from data and integrate this estimate into \colsafe, similar as~\cite{tokmak2024pacsbo} has done for \safeopt.
\section{Proofs}
\label{sec:proof}
Here, we provide the proofs for Lemmas~\ref{lem:bounds} and~\ref{lem:bound_uncertainty} as well as for Theorem~\ref{thm:safety}.

\subsection{Proof of Lemma~\ref{lem:bounds}}
\label{sec:proof_lem_1}

Similar to~\cite{abbasi2011online}, we derive the bounds for the Nadaraya Watson estimator using a super-martingale, which we derive in the next lemma.
\begin{lem}
    \label{lem:super-martingale}
    Under Assumption~\ref{ass:observation_model}, let $(p_n)_{n\in\mathbb{N}}$ be a scalar predictable process with respect to $(\mathcal{F}_n)$, and define, for every $\xi\in\R^m$,
    \[
        \nu_n(\xi)\coloneqq \exp\left[\sum_{r=1}^n\left(\frac{\xi^\transp\omega_np_n}{\sigma} - \frac{1}{2}\xi^\transp\xi p_n^2\right)\right].
    \]
    Then, for every $\xi\in\R^m$, $(\nu(\xi))_{n\in\mathbb{N}}$ is a non-negative super-martingale with respect to $(\mathcal{F}_n)$.
    In particular, $\E[\nu_n(\xi)]\!\le\! 1$.
\end{lem}
\begin{proof}
    Fix some $\xi\in\R^m$, and let
    \[
        E_n\coloneqq \exp\left(\frac{\xi^\transp \omega_n p_n}{\sigma} - \frac{1}{2}\xi^\transp\xi p_n^2\right),\quad n\in\mathbb{N}.
    \]
    Clearly, $\nu(\xi)=E_1E_2\cdots E_n\ge0$, where $E_r$ is $\mathcal{F}_r$-measurable.
    Note that, for all $n\in\mathbb{N}$,
    \begin{align*}
        \E[E_n\mid \mathcal{F}_{n-1}] &= \E\left[\exp\left(\frac{\xi^\transp\omega_n p_n}{\sigma}\right)\exp\left(-\frac{1}{2}\xi^\transp\xi p_n^2\right)\middle|\mathcal{F}_{n-1}\right]\\
        &= \E\left[\exp\left(\frac{\xi^\transp\omega_np_n}{\sigma}\right)\middle|\mathcal{F}_{n-1}\right]\exp\left(-\frac{1}{2}\xi^\transp\xi p_n^2\right),
    \end{align*}
    because $p_n$ is $\mathcal{F}_{n-1}$-measurable.
    Let $\gamma_n = \xi p_n / \sigma$ for all $n\in\mathbb{N}$; note that $(\gamma_n)$ is predictable with respect to $(\mathcal{F}_n)$.
    Then, because we assume sub-Gaussian noise, we have
    \[
        \E\left[\exp\left(\frac{\xi^\transp\omega_n p_n}{\sigma}\right)\middle|\mathcal{F}_{n-1}\right]\le \exp\left(\frac{\gamma_n^\transp\gamma_n\sigma^2}{2}\right),
    \]
    and, therefore,
    \[
        \E[E_n\mid\mathcal{F}_{n-1}] \le \exp\left(\frac{\gamma_n^\transp\gamma_n\sigma^2}{2}\right)\exp\left(-\frac{\gamma_n^\transp\gamma_n\sigma^2}{2}\right) = 1.
    \]
    Next, for every $n>1$,
    \begin{align*}
        \E[\nu(\xi)\mid\mathcal{F}_{n-1}] &= \E[E_1\cdots E_{n-1}E_n\mid\mathcal{F}_{n-1}]\\
        &= E_1\cdots E_{n-1}\E[E_n\mid\mathcal{F}_{n-1}]\\
        &\le \nu_{n-1}(\xi).
    \end{align*}
    Thus, using the tower property of conditional expectations (see \cite[Theorem~9.7(i)]{Williams-91}),
    \begin{align*}
        \E[\nu_n(\xi)] &= \E[\E[\nu_n(\xi)\mid\mathcal{F}_{n-1}]]\\
        &\le\E[\nu_{n-1}(\xi)]\le\cdots\le \E[\nu_1(\xi)]\\ 
        &=\E[\E[E_1\mid\mathcal{F}_0]]\\
        &\le 1,
    \end{align*}
    which completes the proof.
\end{proof}

For bounds that hold uniformly with respect to time, we will make a stopping time argument, similar to~\cite{abbasi2011online}.
This requires us to extend the previous result to the case of stopping times.
\begin{cor}
    \label{cor:stopping_times}
    Under the conditions of Lemma~\ref{lem:super-martingale}, let $\tau\in\mathbb{N}\cup\{+\infty\}$ be a stopping time with respect to $(\mathcal{F}_n)$.
    Then, $\nu_\tau(\xi)$ exists almost surely (in $\mathcal{F}_\infty$), where $\nu_\tau(\xi)$ is defined in the general case (when possibly $\mathbb{P}[\tau=\infty]>0$) as the almost sure limit $\lim_{n\to\infty}\nu_{\tau\land n}(\xi)$.
    Furthermore, $\E[\nu_\tau(\xi)]\le 1$ for every $\xi\in\R^m$.
\end{cor}
\begin{proof}
    Fix a $\xi\in\R^m$, and let $\tilde{\nu}_n(\xi)\coloneqq \nu_{\tau\land n}(\xi)$ for all $n\in\mathbb{N}$.
    Then, since, by Lemma~\ref{lem:super-martingale}, $(\nu_n(\xi))$ is a non-negative super-martingale, $(\tilde{\nu}_n(\xi))_{n\in\mathbb{N}}$ is also a non-negative super-martingale (see \cite[Section~10.9]{Williams-91}).
    This implies that $\E[\nu_{\tau\land n}(\xi)]\le\E[\nu_{\tau\land 1}(\xi)]\le 1$.
    Furthermore, by the convergence theorem for non-negative super-martingales (see \cite[Corollary~11.7]{Williams-91}), $\nu_\tau(\xi) = \lim_{n\to\infty}\nu_{\tau\land n}(\xi)$ exists almost surely.
    Moreover, by Fatou's lemma (see \cite[Lemma~5.4]{Williams-91}), $\E[\nu_\tau(\xi)] = \E[\lim\inf_{n\to\infty}\nu_{\tau\land n}(\xi)]\le\lim\inf_{n\to\infty}\E[\nu_{\tau\land n}(\xi)]\le 1$.
\end{proof}

\begin{lem}
\label{lem:tech_a}
    Under Assumption~\ref{ass:observation_model}, let $(p_n)_{n\in\mathbb{N}}$ be a real scalar predictable process with respect to $(\mathcal{F}_n)$.
    Further, let $\tau\in\mathbb{N}\cup\{\infty\}$ be a stopping time with respect to $(\mathcal{F}_n)$.
    Define
    \begin{align*}
        s_n &\coloneqq \sum_{r=1}^np_r\omega_r\\
        v_n &\coloneqq \sum_{r=1}^np_r^2,
    \end{align*}
    where $n\in\mathbb{N}$.
    Then, for every $0<\delta<1$, with probability at least $1-\delta$,
    \[
        s_\tau^\transp s_\tau \le 2\sigma^2\ln\left[\frac{(v_\tau+1)^{\frac{m}{2}}}{\delta}\right](v_\tau+1).
    \]
\end{lem}
\begin{proof}
    Without loss of generality, let $\sigma=1$.
    For every $\xi\in\R^m$, let
    \[
        \nu_n(\xi)\coloneqq\exp\left(\xi^\transp s_n - \frac{1}{2}\xi^\transp\xi v_n\right).
    \]
    From Corollary~\ref{cor:stopping_times}, we note that, for every $\xi\in\R^m$, $\E[\nu_\tau(\xi)]\le 1$.
    Let now $\Xi$ be a $\mathcal{N}(0, I)$ random vector, independent of all other variables.
    Clearly, $\E[\nu_\tau(\Xi)\mid\Xi]\le 1$.
    Define
    \[
        \nu_\tau\coloneqq \E[\nu_\tau(\Xi)\mid\mathcal{F}_\infty].
    \]
    Then, $\E[\nu_\tau]\le 1$, since, by the tower property of conditional expectations (see \cite[Theorem~9.7(i)]{Williams-91}),
    \begin{align*}
        \E[\nu_\tau] &= \E[\E[\nu_\tau(\Xi)\mid\mathcal{F}_\infty]]\\
        &= \E[\nu_\tau(\Xi)]\\
        &= \E[\E[\nu_\tau(\Xi)\mid\Xi]]\\
        &\le 1.
    \end{align*}
    We can also express $\nu_\tau$ directly as
    \begin{align*}
        \nu_\tau &= \frac{1}{(2\pi)^{\frac{m}{2}}}\int_{\R^m}\exp\left(\xi^\transp s_\tau - \frac{1}{2}\xi^\transp\xi v_\tau\right)\exp\left(-\frac{1}{2}\xi^\transp\xi\right)\diff\xi\\
        &= \frac{1}{(2\pi)^{\frac{m}{2}}}\int_{\R^m} \exp\left(\xi^\transp s_\tau - \frac{1}{2}\xi^\transp\xi v_\tau - \frac{1}{2}\xi^\transp\xi\right)\diff\xi\\
        &= \frac{1}{(2\pi)^{\frac{m}{2}}}\int_{\R^m} \exp\left(\xi^\transp s_\tau - \frac{1}{2}(v_\tau + 1)\xi^\transp\xi\right)\diff\xi\\
        &= \frac{1}{(2\pi)^{\frac{m}{2}}}\int_{\R^m}  \exp\left(\xi^\transp s_\tau - \frac{1}{2}(v_\tau + 1)\xi^\transp\xi \right.\\
        &\qquad\qquad\qquad\qquad\; +\left. \frac{1}{2}\frac{1}{v_\tau + 1}s_\tau^\transp s_\tau - \frac{1}{2}\frac{1}{v_\tau + 1}s_\tau^\transp s_\tau\right)\diff\xi\\
        &= \frac{1}{(2\pi)^{\frac{m}{2}}}\int_{\R^m}\exp\left(-\frac{1}{2}(v_\tau+1)\left\|\xi-\frac{1}{v_\tau+1}s_\tau\right\|_2^2\right)\\
        &\qquad\qquad\qquad\;\cdot \exp\left(\frac{1}{2}\frac{1}{v_\tau+1}s_\tau^\transp s_\tau\right)\diff\xi\\
        &=(v_\tau+1)^{-\frac{m}{2}}\exp\left(\frac{1}{2}\frac{1}{v_\tau+1}s_\tau^\transp s_\tau\right).
    \end{align*}
    The last step follows as 
    \begin{align*}
\left(\frac{v_r+1}{2 \pi}\right)^\frac{m}{2} \exp\left( - \frac{1}{2} (v_\tau + 1) \left\| \xi - \frac{1}{v_\tau + 1} s_\tau \right\|_2^2 \right)
\end{align*}
is the probability density function of an $m$-dimensional normal distribution $\displaystyle \mathcal{N}\left( \frac{1}{v_\tau + 1} s_\tau, \frac{1}{v_\tau + 1} \right)$, so its integral over $\mathbb{R}^m$ is equal to $1$.

    Therefore, $\mathbb{P}[\delta\nu_\tau > 1]$ is equal to
    \begin{align*}
        &\mathbb{P}\left[\frac{\delta}{(v_\tau+1)^{\frac{m}{2}}}\exp\left(\frac{1}{2}\frac{s_\tau^\transp s_\tau}{v_\tau+1}\right) > 1\right]\\
        &\qquad =\mathbb{P}\left[\exp\left(\frac{1}{2}\frac{s_\tau^\transp s_\tau}{v_\tau+1}\right) > \frac{(v_\tau+1)^{\frac{m}{2}}}{\delta}\right]\\
        &\qquad = \mathbb{P}\left[s_\tau^\transp s_\tau > 2\ln\left(\frac{(v_\tau+1)^{\frac{m}{2}}}{\delta}\right)(v_\tau+1)\right].
    \end{align*}
    Recall now that $\E[\nu_\tau]\le 1$.
    Hence, due to Markov's inequality (see \cite[Theorem~6.4]{Williams-91}),
    \[
        \mathbb{P}[\delta\nu_\tau > 1]\le \delta\E[\nu_\tau]\le \delta,
    \]
    which completes the proof.
\end{proof}
\begin{remark}
    We can here set $\sigma=1$ without loss of generality because it is a proxy standard deviation of $\omega_i$. Thus, we can assume that $\sigma = 1$ by working with the scaled variables $\tilde{\omega}_i = \omega_i / \sigma$, which are conditionally $1$-sub Gaussian, and also defining $\tilde s_n = s_n / \sigma$. Once Lemma~\ref{lem:tech_a} is established for these variables, one can revert back to $\omega_i$ and $s_n$ by reintroducing $\sigma$.
\end{remark}

\begin{cor}
    \label{cor:tech_a}
    Under the same assumptions as in Lemma~\ref{lem:tech_a}, for every $0<\delta<1$, it holds with probability at least $1-\delta$ that
    \[
        s_n^\transp s_n \le 2\sigma^2\ln\left(\frac{(v_n+1)^{\frac{m}{2}}}{\delta}\right)(v_n+1)
    \]
    for all $n\in\mathbb{N}$.
\end{cor}
\begin{proof}
    Let $\tau = \min\{n\in\mathbb{N}\colon s_n^\transp s_n>2\sigma^2\ln[(v_n+1)^{m/2}/\delta](v_n+1)\}$, and $\tau=\infty$ if $s_n^\transp s_n\le 2\sigma^2 \ln[(v_n+1)^{m/2} / \delta](v_n+1)$ for all $n$.
    It is clear that $\tau$ is a stopping time in $\mathbb{N}\cup\{+\infty\}$.
    Further, from the definition of $\tau$, we have that $\tau<\infty$ if, and only if, $s_n^\transp s_n > s\sigma^2\ln[(v_n+1)^{m/2}/\delta](v_n+1)$ for some $n\in\mathbb{N}$.
    Thus, the corollary will be established if we prove that $\mathbb{P}[\tau < \infty]\le \delta$.

    Now, by Lemma~\ref{lem:tech_a},
    \begin{align*}
        \mathbb{P}[\tau<\infty] &= \mathbb{P}\left[s_\tau^\transp s_\tau > 2\sigma^2\ln\left(\frac{(v_\tau+1)^{\frac{m}{2}}}{\delta}\right)(v_\tau+1), \tau<\infty\right]\\
        &\le \mathbb{P}\left[s_\tau^\transp s_\tau > 2\sigma^2\ln\left(\frac{(v_\tau+1)^{\frac{m}{2}}}{\delta}\right)(v_\tau+1)\right]\\
        &\le \delta.
    \end{align*}
    This concludes the proof.
\end{proof}

With these technical preliminaries, we can finally prove Lemma~\ref{lem:bounds}.
We begin with the observation that
\begin{multline}
\label{eqn:begin_obs}
    \norm{\mu_n(a, i) - h(a,i)} = \left\lVert\sum\limits_{r=1}^n\frac{K_\lambda(a, a_r)}{\kappa_n(a)}\hat{h}_r - h(a, i)\right\lVert_2 \\
    \le \sum\limits_{r=1}^n\theta_r\norm{h(a_n,i) - h(a, i)} + \left\lVert\sum\limits_{r=1}^n\theta_r\omega_r\right\rVert_2,
\end{multline}
where $\theta_r\coloneqq\frac{K_\lambda(a, a_r)}{\kappa_n(a)}$.
Note, that $\sum_{r=1}^n\theta_r = 1$.
Due to Assumption~\ref{ass:kernel}, if $K_\lambda(a, a_n)>1$, then $\frac{\norm{a-a_n}}{\lambda}\le 1$.
Therefore, see Assumption~\ref{ass:smoothness_assumption}, if $K_\lambda(a, a_n)>1$, then
\[
    \norm{h(a_n) - h(a)}\le L\norm{a_n - a} \le L\lambda,
\]
and, since the weights $\theta_r$ sum up to 1,
\[
    \sum\limits_{r=1}^n\theta_r\norm{h(a_n, i) - h(a, i)}\le L\lambda.
\]
For the last term in~\eqref{eqn:begin_obs}, observe that
\begin{equation}
\label{eqn:next_obs}
    \left\lVert\sum\limits_{r=1}^n\theta_r\omega_r\right\rVert = \frac{1}{\kappa_n(a)}\left\lVert\sum\limits_{r=1}^nK_\lambda(a, a_r)\omega_r\right\rVert_2.
\end{equation}
According to Corollary~\ref{cor:tech_a}, with $p_r = K_\lambda(a, a_r)$, the right-hand side of~\eqref{eqn:next_obs} is upper bounded, with probability at least $1-\delta$, for all $n\in\mathbb{N}$, by
\begin{align*}
    &\frac{1}{\kappa_n(a)}\sigma \cdot \\
    &\sqrt{2\ln\left(\frac{1}{\delta}\left(1+\sum\limits_{r=1}^nK_\lambda^2(a, a_r)\right)^{\frac{m_i}{2}}\right)\left(1+\sum\limits_{r=1}^nK_\lambda^2(a, a_r)\right)}.
\end{align*}
Furthermore, since $K_\lambda(a, a_n)\le 1$ (see Assumption~\ref{ass:kernel}), we obtain
\begin{multline*}
    \frac{1}{\kappa_n(a)}\left\lVert\sum\limits_{r=1}^n K_\lambda(a, a_r)\omega_r\right\rVert_2 \\
    \le \sigma\sqrt{2\ln\left(\frac{(1+\kappa_n(a))^{\frac{m_i}{2}}}{\delta}\right)}\frac{\sqrt{1+\kappa_n(a)}}{\kappa_n(a)}.
\end{multline*}
Observe next that, if $\kappa_n(a) > 1$,
\[
    \frac{\sqrt{1+\kappa_n(a)}}{\kappa_n(a)} < \frac{\sqrt{2\kappa_n(a)}}{\kappa_n(a)} = \frac{\sqrt{2}}{\sqrt{\kappa_n(a)}}.
\]
Therefore, with probability at least $1-\delta$, for all $n\in\mathbb{N}$ such that $\kappa_n(a) > 1$,
\begin{multline*}
    \frac{1}{\kappa_n(a)}\left\lVert\sum\limits_{r=1}^nK_\lambda(a, a_r)\omega_r\right\rVert_2\\
    \le \frac{2\sigma}{\kappa_n(a)}\sqrt{\kappa_n(a)\ln\left(\frac{(1+\kappa_n(a))^{\frac{m_i}{2}}}{\delta}\right)},
\end{multline*}
whereas for all $n\in\mathbb{N}$ such that $0<\kappa_n(a)\le 1$,
\begin{align*}
    &\frac{1}{\kappa_n(a)}\abs{\sum\limits_{r=1}^nK_\lambda(a, a_r)\omega_r}\\
    &\qquad \le \frac{\sigma}{\kappa_n(a)}\sqrt{2\ln\left(\frac{(1+\kappa_n(a))^{\frac{m_i}{2}}}{\delta}\right)}\sqrt{1+\kappa_n(a)}\\
    &\qquad \le \frac{2\sigma}{\kappa_n(a)}\sqrt{\frac{\ln\left(2^{\frac{m_i}{2}}\right)}{\delta}}.
\end{align*}

For each $a\in\mathcal{A}$, the result holds uniformly on $n\in\mathbb{N}$ with probability at least $1-\delta/D$, so the proof is completed by a union bound over $a\in\mathcal{A}$.

\subsection{Proof of Lemma~\ref{lem:bound_uncertainty}}
\label{sec:proof_lem_2}

Following Assumption \ref{ass:kernel}, we note that for every $a_n$ in the support of the kernel,
\[
    \chi_\mathrm{K}\le K_\lambda(a, a_n).
\]
Hence, for a fixed value of $\delta$, we have
\[
    \frac{\sqrt{\ln \left( \frac{2^{\frac{1}{2}}}{\delta}D \right)}}{\kappa_n(a)}\le \frac{\sqrt{\ln \left( \frac{2^{\frac{1}{2}}}{\delta}D \right)}}{N_{\bar{\beta}}\chi_\mathrm{K}}
\]
if $\kappa_n(a)\le 1$, and
\[
    \frac{\sqrt{\kappa_n(a) \ln \left( \frac{(1+\kappa_n(a))^{\frac{1}{2}}}{\delta}D \right)}}{\kappa_n(a)}\le \frac{\sqrt{N_{\bar{\beta}} \chi_\mathrm{K} \ln \left( \frac{(1+N_{\bar{\beta}}\chi_\mathrm{K})^{\frac{1}{2}}}{\delta}D \right)}}{N_{\bar{\beta}}\chi_\mathrm{K}}
\]
if $\kappa_n(a) > 1$.
From this, the result follows.

\subsection{Proof of Theorem~\ref{thm:safety}}
\label{sec:proof_thm_1}

We collect all variables and operators relevant to the proof of Theorem~\ref{thm:safety} in \tabref{tab:notation}.

\begin{table}
    \centering
    \caption{Summary of main variables and operators.}
    \label{tab:notation}
    \begin{tabular}{ll}
        \toprule
        $a$ & Policy parameters\\
        $f$ & Reward function \\
        $g_i$ & Constraint function $i$\\
        $h_i$ & Surrogate selector function for index $i$\\
        $\mathcal{I}$ & Set of all indices\\
        $\mathcal{I}_\mathrm{g}$ & Set of all indices pertaining to constraints (\ie $i>0$)\\
        $\mathcal{A}$ & Parameter space\\
        $S_n$ & Safe set at iteration $n$\\
        $Q_n$ & Confidence intervals at iteration $n$\\
        $C_n$ & Contained set at iteration $n$\\
        $M_n$ & Set of potential maximizers at iteration $n$\\
        $G_n$ & Set of potential expanders at iteration $n$\\
        $u_n(a,i)$ & Upper bound of the confidence interval at iteration $n$ \\
        $l_n(a, i)$ & Lower bound of the confidence interval at iteration $n$\\
        $w_n(a,i)$ & $\coloneqq u_n(a,i) - l_n(a,i)$ \\
        $\bar{\beta}$ & Uncertainty upper bound\\
        $R_{\bar{\beta}}$ & Reachability operator\\
        $\bar{R}_{\bar{\beta}}$ & Closure of $R_{\bar{\beta}}$\\
        $L$ & Lipschitz constant\\
        \bottomrule
    \end{tabular}
\end{table}

Given a value of $\bar{\beta} > L \lambda$, let $N_{\bar{\beta}} \geq 1$ be defined as in the paragraph following Lemma~\ref{lem:bound_uncertainty}.

The proof of Theorem~\ref{thm:safety} basically involves three steps: we need to show \emph{(i)} safety, \emph{(ii)} that the algorithm stops after a finite $n$, and \emph{(iii)} that after those finite steps, we are close to the optimum.
As the algorithm is similar to \safeopt, the proof strategy follows the one proposed in~\cite{berkenkamp2021bayesian}.

\subsubsection{Safety}
Before providing the proof, let us state again the safety result.   
If $h$ is $L$-Lipschitz-continuous, then, with probability at least $1-\delta$, for all $n\ge 0$, $i\in\mathcal{I}_\mathrm{g}$ and $a\in S_n$, the following holds:
\[
    \norm{g(a, i)}\ge c_i.
\]

We will prove this by induction.
For the base case $n=0$, by definition, for any $a\in S_0$ and $i\in\mathcal{I}_\mathrm{g}$, $\norm{g_i(a)}\ge c_i$.

For the induction step, assume that for some $n\ge 1$, for any $a\in S_{n-1}$ and for all $i\in\mathcal{I}_\mathrm{g}$, $\norm{g_i(a)}\ge c_i$.
Then, for any $a\in S_n$, by definition, for all $i\in\mathcal{I}_\mathrm{g}$, there exists a $z_i\in S_{n-1}$ such that
\begin{align*}
    c_i &\le l_n(z_i, i) - L\norm{z_i - a}\\
    &\le \norm{g(z_i, i)} - L\norm{z_i - a}\tag*{by Lemma~\ref{lem:bounds}}\\
    &\le \norm{g(a, i)} \tag*{by Lipschitz-continuity.}
\end{align*}

\subsubsection{Finite iterations}
To prove that the algorithm stops after finite time, we first require some properties of the sets and functions we defined before.
Most importantly, the upper confidence bounds are decreasing, and the lower confidence bounds are increasing with the number of iterations since we have $C_{n+1}\supseteq C_n$ for all iterations $n$.
As the set definitions are the same as in~\cite{berkenkamp2021bayesian}, the proof follows as shown therein.
\begin{lem}[{\cite[Lemma~7.1]{berkenkamp2021bayesian}}]
\label{lem:set_properties}
    The following hold for any $n\ge 1$:
    \begin{itemize}
        \item[(i)] $\forall a\in\mathcal{A}$, $\forall i\in\mathcal{I}$, $u_{n+1}(a,i)\le u_n(a,i)$,
        \item[(ii)] $\forall a\in\mathcal{A}$, $\forall i\in\mathcal{I}$, $l_{n+1}(a,i)\ge l_n(a,i)$,
        \item[(iii)] $\forall a\in\mathcal{A}$, $\forall i\in\mathcal{I}$, $w_{n+1}(a,i)\le w_n(a,i)$,
        \item[(iv)] $S_{n+1}\supseteq S_n\supseteq S_0$,
        \item[(v)] $S\subseteq R \implies R_{\bar{\beta}} (S)\subseteq R_{\bar{\beta}}(R)$,
        \item[(vi)] $S\subseteq R \implies \bar{R}_{\bar{\beta}} (S)\subseteq \bar{R}_{\bar{\beta}}(R)$.
    \end{itemize}
\end{lem}

Using these results, we first show that the safe set eventually stops expanding.
\begin{lem}
\label{lem:final_Sn}
There exists an $n^* \ge 0$ (possibly random, but always finite) such that, for all $n > n^*$, $S_n = S_{n^*}$.
\end{lem}
\begin{proof}
By Lemma~\ref{lem:set_properties} (iv), the sets $S_n$ are increasing in $n$, and $S_n \subseteq \mathcal{A}$ for all $n$, so $\bar{S} := \cup_{n=0}^\infty S_n \subseteq \mathcal{A}$ is finite. Furthermore, each $a \in \bar{S}$ belongs to $S_n$ for some $n$; let $n(a)$ be the smallest such $n$. Then, $n^* := \max_{a \in \bar{S}} n(a)$ is finite (because $\bar{S}$ is finite) and satisfies $S_{n^*} = \bar{S}$, so $S_n = S_{n^*}$ for all $n > n^*$.
\end{proof}
Next, note that the sets of expanders and maximizers are contained in each other as well once the safe set does not increase, as stated in the following lemma, which can be proven as in~\cite{berkenkamp2021bayesian}.
\begin{lem}[{\cite[Lemma~7.2]{berkenkamp2021bayesian}}]
\label{lem:max_exp}
    For any $n_1\ge n_0\ge 1$, if $S_{n_1}=S_{n_0}$, then, for any $n$ such that $n_0\le n < n_1$,
    \[
        G_{n+1}\cup M_{n+1} \subseteq G_n \cup M_n.
    \]
\end{lem}

When running \colsafe, we repeatedly choose the most uncertain element from $G_n$ and $M_n$.
Since these sets are contained in each other, if the safe set does not expand, we gain more information about these sets with each sample.
Further, after a finite number of iterations of \colsafe, we can guarantee that the uncertainty of each point in these sets is upper bounded as shown in Lemma~\ref{lem:bound_uncertainty}.
\begin{lem}
\label{cor:limit_uncertainty}
    For any $n\ge n^*$ 
    (where $n^*$ was defined in Lemma~\ref{lem:final_Sn}), we have for any $a\in G_{n+N_{\bar{\beta}} |\mathcal{A}|}\cup M_{n+N_{\bar{\beta}} |\mathcal{A}|}$ and $i\in\mathcal{I}$ that
    \[
        w_{n+N_{\bar{\beta}} |\mathcal{A}|}(a,i) \le 2\bar{\beta}.
    \]
\end{lem}
\begin{proof}
By Lemma~\ref{lem:max_exp}, $G_m\cup M_m$ is decreasing with respect to $m$ for $m \ge n$. Also, since the parameter $a_m$ is the most uncertain element of $G_m \cup M_m$, after sampling $N_{\bar{\beta}} |\mathcal{A}|$ times from $G_n \cup M_n$ it follows that each element $a \in G_{n+N_{\bar{\beta}} |\mathcal{A}|}\cup M_{n+N_{\bar{\beta}} |\mathcal{A}|}$ has been sampled at least $N_{\bar{\beta}}$ times, so its uncertainty is at most $2\bar{\beta}$, by Lemma~\ref{lem:bound_uncertainty}. Otherwise, there would be another element $a' \in G_{n+N_{\bar{\beta}} |\mathcal{A}|}\cup M_{n+N_{\bar{\beta}} |\mathcal{A}|}$ that has been selected more than $N_{\bar{\beta}}$ times, so the $(N_{\bar{\beta}}+1)$-th time it was selected its uncertainty was larger than that of $a$ (i.e., larger than $2 \bar{\beta}$). This contradicts Lemma~\ref{lem:bound_uncertainty}, which concludes the proof.
\end{proof}
According to Lemma~\ref{cor:limit_uncertainty}, after at most $N_{\bar{\beta}} |\mathcal{A}|$ steps since the safe sets stop expanding, the most uncertain element within the sets of potential maximizers and potential expanders has an uncertainty of at most $2\bar{\beta}$. This implies that CoLSafe stops after at most $n^* + N_{\bar{\beta}} |\mathcal{A}|$ iterations.

\subsubsection{Optimality}
To prove optimality, we must first show that we fully explore the reachable set.
Given that the reachability operator~\eqref{eqn:reachability} is defined in terms of the accuracy $2 \bar{\beta}$, the safe set has to increase eventually unless it is impossible to do so.
\begin{lem}[{\cite[Lemma~7.4]{berkenkamp2021bayesian}}]
    \label{lem:safe_set_increase_i}
    For any $n\ge 0$, if $\bar{R}_{\bar{\beta}} (S_0)\setminus S_n\neq\varnothing$, then $R_{\bar{\beta}}(S_n)\setminus S_n\neq\varnothing$.
\end{lem}
\begin{lem}
    \label{lem:safe_set_increase_ii}
    With probability at least $1 - \delta$, if, for some $n \geq 1$, $S_{n-1} = S_n$ and $w_n(a,i) \le 2\bar{\beta}$ for all $a\in M_n \cup G_n$ and $i\in\mathcal{I}$, then $\bar{R}_{\bar{\beta}}(S_0)\subseteq S_n$.
\end{lem}
\begin{proof}
    Assume, to the contrary, that $\bar{R}_{\bar{\beta}}(S_0)\setminus S_n \neq \varnothing$.
    By Lemma~\ref{lem:safe_set_increase_i}, we have that $R_{\bar{\beta}}(S_n)\setminus S_n\neq\varnothing$.
    Equivalently, by definition, for all $i\in\mathcal{I}_\mathrm{g}$:
\begin{equation}
        \label{eqn:safe_set_increase}
        \exists a\in R_{\bar{\beta}}(S_n)\setminus S_n, \exists z_i\in S_n\colon \norm{g(z_i, i)} - 2\bar{\beta} - L\norm{z_i-a}\ge c_i.
    \end{equation}
    Then, for $a$ and $\{z_i\}$ from \eqref{eqn:safe_set_increase} it holds that, for all $i\in\mathcal{I}_\mathrm{g}$,
    \begin{align*}
        u_n(z_i,i) - L\norm{z_i-a} &\ge \norm{g(z_i, i)} - L\norm{z_i-a}\tag*{by Lemma~\ref{lem:bounds}}\\
        &\ge \norm{g(z_i,i)} - 2\bar{\beta} - L\norm{z_i-a}\\
        &\ge c_i \tag*{by~\eqref{eqn:safe_set_increase}.}
    \end{align*}
    Therefore, by the definition of $G_n$ in \eqref{eqn:expanders}, $z_i\in G_n \subseteq S_n$ for all $i\in\mathcal{I}_\mathrm{g}$.

    Finally, since $S_{n-1}=S_n$, we have that $z_i\in S_{n-1}$ for all $i\in\mathcal{I}_\mathrm{g}$. Then, the assumption that $w_n(z_i,i) \le 2\bar{\beta}$ for all $i\in\mathcal{I}_\mathrm{g}$ implies that, for all $i\in\mathcal{I}_\mathrm{g}$,
    \begin{align*}
        &l_n(z_i,i) - L\norm{z_i-a} \\
        &\quad \ge \norm{g(z_i,i)} - w_n(z_i,i) - L\norm{z_i - a} \tag*{by Lemma~\ref{lem:bounds}}\\
        &\quad \ge \norm{g(z_i,i)} - 2\bar{\beta} - L\norm{z_i - a} \\
        &\quad \ge c_i \tag*{by~\eqref{eqn:safe_set_increase}.}
    \end{align*}
    This means that $a\in S_n$, which is a contradiction.
\end{proof}

From the previous result, we get that once \colsafe stops, it will have explored the safely reachable set $\bar{R}_{\bar{\beta}}(S_0)$ to the desired accuracy. 
From this it follows that the pessimistic estimate in~\eqref{eqn:optimum} is also $2\bar{\beta}$-close to the optimum value within the safely reachable set, $\bar{R}_{\bar{\beta}}(S_0)$.
\begin{lem}
    \label{lem:epsilon_optimal}
    With probability at least $1-\delta$, if, for some $n \geq 1$, $S_{n-1} = S_n$ and $w_n(a,i) \le 2\bar{\beta}$ for all $a\in M_n \cup G_n$ and $i\in\mathcal{I}$, we have that
    \[
        f(\hat{a}_n) \ge \max_{a\in\bar{R}_{\bar{\beta}}(S_0)}f(a) - 2\bar{\beta}.
    \]
\end{lem}
\begin{proof}
    Let $a^*\coloneqq \argmax_{a\in S_n}f(a)$.
    Note that $a^*\in M_n$ since, for all $a \in S_n$,
    \begin{align*}
        u_n(a^*,0)&\ge f(a^*) \tag*{by Lemma~\ref{lem:bounds}}\\
        &\ge f(a) \\
        &\ge l_n(a,0) \tag*{by Lemma~\ref{lem:bounds},}
    \end{align*}
    thus $u_n(a^*,0) \ge \max_{a\in S_n}l_n(a,0)$.
    
    We will next show that $f(\hat{a}_n)\ge f(a^*)-2\bar{\beta}$.
    Assume, to the contrary, that
    \begin{equation}
        \label{eqn:optimality_contrary}
        f(\hat{a}_n) < f(a^*) - 2\bar{\beta}.
    \end{equation}
    Then, we have
    \begin{align*}
        l_n(a^*,0) &\le l_n(\hat{a}_n,0)\\
        &\le f(\hat{a}_n)\tag*{by Lemma~\ref{lem:bounds}}\\
        &< f(a^*)-2\bar{\beta} \tag*{by~\eqref{eqn:optimality_contrary}}\\
        &\le u_n(a^*,0) - 2\bar{\beta} \tag*{by Lemma~\ref{lem:bounds}}\\
        &\le l_n(a^*,0),
    \end{align*}
    where the last inequality follows from the fact that $a^* \in M_n$, so $w_n(a^*, i) \le 2 \bar{\beta}$ for all $i \in \mathcal{I}$.
    This statement is a contradiction.

    Finally, Lemma~\ref{lem:safe_set_increase_ii} implies that $\bar{R}_{\bar{\beta}}(S_0)\subseteq S_n$.
    Therefore,
    \begin{align*}
        \max_{a\in\bar{R}_{\bar{\beta}}(S_0)}f(a)-2\bar{\beta} &\le \max_{a\in S_n}f(a) - 2\bar{\beta} \\
        &= f(a^*) - 2\bar{\beta} \\
        &\le f(\hat{a}_n).
    \end{align*}
\end{proof}
\section*{Acknowledgements}

The authors would like to thank Bhavya Sukhija for support with the robot experiments and simulations.

\bibliographystyle{IEEEtran}
\bibliography{IEEEabrv,ref}

\begin{thebibliography}{10}
\providecommand{\url}[1]{#1}
\csname url@samestyle\endcsname
\providecommand{\newblock}{\relax}
\providecommand{\bibinfo}[2]{#2}
\providecommand{\BIBentrySTDinterwordspacing}{\spaceskip=0pt\relax}
\providecommand{\BIBentryALTinterwordstretchfactor}{4}
\providecommand{\BIBentryALTinterwordspacing}{\spaceskip=\fontdimen2\font plus
\BIBentryALTinterwordstretchfactor\fontdimen3\font minus
  \fontdimen4\font\relax}
\providecommand{\BIBforeignlanguage}[2]{{%
\expandafter\ifx\csname l@#1\endcsname\relax
\typeout{** WARNING: IEEEtran.bst: No hyphenation pattern has been}%
\typeout{** loaded for the language `#1'. Using the pattern for}%
\typeout{** the default language instead.}%
\else
\language=\csname l@#1\endcsname
\fi
#2}}
\providecommand{\BIBdecl}{\relax}
\BIBdecl

\bibitem{duan2016benchmarking}
Y.~Duan, X.~Chen, R.~Houthooft, J.~Schulman, and P.~Abbeel, ``Benchmarking deep
  reinforcement learning for continuous control,'' in \emph{International
  Conference on Machine Learning}, 2016, pp. 1329--1338.

\bibitem{jiang2020learning}
Z.-P. Jiang, T.~Bian, and W.~Gao, ``Learning-based control: A tutorial and some
  recent results,'' \emph{Foundations and Trends{\textregistered} in Systems
  and Control}, vol.~8, no.~3, pp. 176--284, 2020.

\bibitem{brunke2022safe}
L.~Brunke, M.~Greeff, A.~W. Hall, Z.~Yuan, S.~Zhou, J.~Panerati, and A.~P.
  Schoellig, ``Safe learning in robotics: From learning-based control to safe
  reinforcement learning,'' \emph{Annual Review of Control, Robotics, and
  Autonomous Systems}, vol.~5, pp. 411--444, 2022.

\bibitem{hewing2020learning}
L.~Hewing, K.~P. Wabersich, M.~Menner, and M.~N. Zeilinger, ``Learning-based
  model predictive control: Toward safe learning in control,'' \emph{Annual
  Review of Control, Robotics, and Autonomous Systems}, vol.~3, pp. 269--296,
  2020.

\bibitem{sui2015safe}
Y.~Sui, A.~Gotovos, J.~Burdick, and A.~Krause, ``Safe exploration for
  optimization with {Gaussian} processes,'' in \emph{International conference
  on machine learning}, 2015, pp. 997--1005.

\bibitem{berkenkamp2016safe}
F.~Berkenkamp, A.~P. Schoellig, and A.~Krause, ``Safe controller optimization
  for quadrotors with {G}aussian processes,'' in \emph{IEEE International
  Conference on Robotics and Automation}, 2016, pp. 491--496.

\bibitem{berkenkamp2021bayesian}
F.~Berkenkamp, A.~Krause, and A.~P. Schoellig, ``Bayesian optimization with
  safety constraints: safe and automatic parameter tuning in robotics,''
  \emph{Machine Learning}, vol. 112, no.~10, pp. 3713--3747, 2023.

\bibitem{duivenvoorden2017constrained}
R.~R. Duivenvoorden, F.~Berkenkamp, N.~Carion, A.~Krause, and A.~P. Schoellig,
  ``Constrained {B}ayesian optimization with particle swarms for safe adaptive
  controller tuning,'' in \emph{IFAC World Congress}, 2017, pp.
  11\,800--11\,807.

\bibitem{kirschner2019adaptive}
J.~Kirschner, M.~Mutny, N.~Hiller, R.~Ischebeck, and A.~Krause, ``Adaptive and
  safe {B}ayesian optimization in high dimensions via one-dimensional
  subspaces,'' in \emph{International Conference on Machine Learning}, 2019,
  pp. 3429--3438.

\bibitem{baumann2023computationally}
D.~Baumann, K.~Kowalczyk, K.~Tiels, and P.~Wachel, ``A computationally
  lightweight safe learning algorithm,'' in \emph{IEEE Conference on Decision
  and Control}, 2023, pp. 1022--1027.

\bibitem{mockus1978application}
J.~Mockus, V.~Tiesis, and A.~Zilinskas, ``The application of {B}ayesian methods
  for seeking the extremum,'' \emph{Towards Global Optimization}, vol.~2, no.
  117-129, p.~2, 1978.

\bibitem{antonova2017deep}
R.~Antonova, A.~Rai, and C.~G. Atkeson, ``Deep kernels for optimizing
  locomotion controllers,'' \emph{arXiv preprint arXiv:1707.09062}, 2017.

\bibitem{calandra2016bayesian}
R.~Calandra, A.~Seyfarth, J.~Peters, and M.~P. Deisenroth, ``Bayesian
  optimization for learning gaits under uncertainty,'' \emph{Annals of
  Mathematics and Artificial Intelligence}, vol.~76, no. 1-2, pp. 5--23, 2016.

\bibitem{marco2016automatic}
A.~Marco, P.~Hennig, J.~Bohg, S.~Schaal, and S.~Trimpe, ``Automatic {LQR}
  tuning based on {G}aussian process global optimization,'' in \emph{IEEE
  International Conference on Robotics and Automation}, 2016, pp. 270--277.

\bibitem{hernandez2016general}
J.~M. Hern\'{a}ndez-Lobato, M.~A. Gelbart, R.~P. Adams, M.~W. Hoffman, and
  Z.~Ghahramani, ``A general framework for constrained {B}ayesian optimization
  using information-based search,'' \emph{The Journal of Machine Learning
  Research}, vol.~17, no.~1, pp. 5549--5601, 2016.

\bibitem{Gelbart2014}
M.~A. Gelbart, J.~Snoek, and R.~P. Adams, ``Bayesian optimization with unknown
  constraints,'' in \emph{Conference on Uncertainty in Artificial
  Intelligence}, 2014, pp. 250--259.

\bibitem{gardner2014bayesian}
J.~R. Gardner, M.~J. Kusner, Z.~E. Xu, K.~Q. Weinberger, and J.~P. Cunningham,
  ``Bayesian optimization with inequality constraints,'' in \emph{International
  Conference on Machine Learning}, 2014, pp. 937--945.

\bibitem{gramacy2011opti}
R.~B. Gramacy and H.~Lee, ``Optimization under unknown constraints,''
  \emph{Bayesian Statistics 9}, 2011.

\bibitem{schonlau1998global}
M.~Schonlau, W.~J. Welch, and D.~R. Jones, ``Global versus local search in
  constrained optimization of computer models,'' \emph{Lecture Notes-Monograph
  Series}, pp. 11--25, 1998.

\bibitem{picheny2014stepwise}
V.~Picheny, ``A stepwise uncertainty reduction approach to constrained global
  optimization,'' in \emph{International Conference on Artificial Intelligence
  and Statistics}, 2014, pp. 787--795.

\bibitem{marco2021robot}
A.~Marco, D.~Baumann, M.~Khadiv, P.~Hennig, L.~Righetti, and S.~Trimpe, ``Robot
  learning with crash constraints,'' \emph{IEEE Robotics and Automation
  Letters}, 2021.

\bibitem{baumann2021gosafe}
D.~Baumann, A.~Marco, M.~Turchetta, and S.~Trimpe, ``Go{S}afe: Globally optimal
  safe robot learning,'' in \emph{IEEE International Conference on Robotics and
  Automation}, 2021, pp. 4452--4458.

\bibitem{sukhija2023gosafeopt}
B.~Sukhija, M.~Turchetta, D.~Lindner, A.~Krause, S.~Trimpe, and D.~Baumann,
  ``Go{S}afe{O}pt: Scalable safe exploration for global optimization of
  dynamical systems,'' \emph{Artificial Intelligence}, vol. 320, p. 103922,
  2023.

\bibitem{wachi2018safe}
A.~Wachi, Y.~Sui, Y.~Yue, and M.~Ono, ``Safe exploration and optimization of
  constrained {MDP}s using {G}aussian processes,'' in \emph{AAAI Conference on
  Artificial Intelligence}, 2018, pp. 6548--6555.

\bibitem{konig2021applied}
C.~König, M.~Turchetta, J.~Lygeros, A.~Rupenyan, and A.~Krause, ``Safe and
  efficient model-free adaptive control via {B}ayesian optimization,'' in
  \emph{IEEE International Conference on Robotics and Automation}, 2021, pp.
  9782--9788.

\bibitem{sui2018stagewise}
Y.~Sui, V.~Zhuang, J.~Burdick, and Y.~Yue, ``Stagewise safe {Bayesian}
  optimization with {Gaussian} processes,'' in \emph{International Conference
  on Machine Learning}, 2018, pp. 4781--4789.

\bibitem{schuster1979contributions}
E.~Schuster and S.~Yakowitz, ``Contributions to the theory of nonparametric
  regression, with application to system identification,'' \emph{The Annals of
  Statistics}, pp. 139--149, 1979.

\bibitem{juditsky1995nonlinear}
A.~Juditsky, H.~Hjalmarsson, A.~Benveniste, B.~Delyon, L.~Ljung,
  J.~Sj{\"o}berg, and Q.~Zhang, ``Nonlinear black-box models in system
  identification: Mathematical foundations,'' \emph{Automatica}, vol.~31,
  no.~12, pp. 1725--1750, 1995.

\bibitem{ljung2006some}
L.~Ljung, ``Some aspects on nonlinear system identification,'' in \emph{IFAC
  Symposium on Identification and System Parameter Estimation}, 2006, pp.
  553--564.

\bibitem{mzyk2020wiener}
G.~Mzyk and P.~Wachel, ``Wiener system identification by input injection
  method,'' \emph{International Journal of Adaptive Control and Signal
  Processing}, vol.~34, no.~8, pp. 1105--1119, 2020.

\bibitem{bottero2022information}
A.~Bottero, C.~Luis, J.~Vinogradska, F.~Berkenkamp, and J.~R. Peters,
  ``Information-theoretic safe exploration with {G}aussian processes,'' in
  \emph{Advances in Neural Information Processing Systems}, 2022, pp.
  30\,707--30\,719.

\bibitem{wood1996estimation}
G.~Wood and B.~Zhang, ``Estimation of the {L}ipschitz constant of a function,''
  \emph{Journal of Global Optimization}, vol.~8, pp. 91--103, 1996.

\bibitem{srinivas2012gaussian}
N.~Srinivas, A.~Krause, S.~M. Kakade, and M.~Seeger, ``Information-theoretic
  regret bounds for {Gaussian} process optimization in the bandit setting,''
  \emph{{IEEE} Trans. Inf. Theory}, vol.~58, no.~5, pp. 3250--3265, 2012.

\bibitem{siciliano2008springer}
B.~Siciliano and O.~Khatib, \emph{Springer Handbook of Robotics}.\hskip 1em
  plus 0.5em minus 0.4em\relax Springer Berlin, Heidelberg, 2008.

\bibitem{davies2015effective}
A.~J. Davies, ``Effective implementation of {G}aussian process regression for
  machine learning,'' Ph.D. dissertation, University of Cambridge, 2015.

\bibitem{liu2020gaussian}
H.~Liu, Y.-S. Ong, X.~Shen, and J.~Cai, ``When {G}aussian process meets big
  data: A review of scalable {GP}s,'' \emph{IEEE Transactions on Neural
  Networks and Learning Systems}, vol.~31, no.~11, pp. 4405--4423, 2020.

\bibitem{rao1996pac}
N.~S. Rao and V.~A. Protopopescu, ``On {PAC} learning of functions with
  smoothness properties using feedforward sigmoidal networks,''
  \emph{Proceedings of the IEEE}, vol.~84, no.~10, pp. 1562--1569, 1996.

\bibitem{tokmak2024pacsbo}
A.~Tokmak, T.~B. Sch{\"o}n, and D.~Baumann, ``{PACSBO}: Probably approximately
  correct safe {B}ayesian optimization,'' in \emph{Symposium on System Theory
  in Data and Optimization}, 2024.

\bibitem{abbasi2011online}
Y.~Abbasi-Yadkori, D.~P{\'a}l, and C.~Szepesv{\'a}ri, ``Online least squares
  estimation with self-normalized processes: An application to bandit
  problems,'' \emph{arXiv preprint arXiv:1102.2670}, 2011.

\bibitem{Williams-91}
D.~Williams, \emph{Probability with Martingales}.\hskip 1em plus 0.5em minus
  0.4em\relax Cambridge University Press, 1991.

\end{thebibliography}

\begin{IEEEbiography}[{\includegraphics[width=1in,height=1.25in,clip,keepaspectratio]{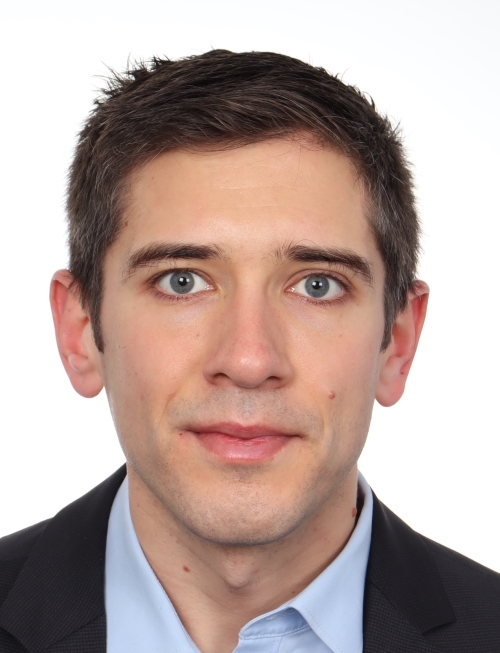}}]{Dominik Baumann}(Member, IEEE) received the Dipl.-Ing.\ degree in electrical engineering from the Dresden University of Technology, Dresden, Germany, in 2016 and the Ph.D.\ degree in electrical engineering from KTH Stockholm, Stockholm, Sweden, in 2020.
He is currently an Assistant Professor with Aalto University, Espoo, Finland. 
He was a joint Ph.D.\ student with the Max Planck Institute for Intelligent Systems, Stuttgart/Tübingen, Germany, and KTH Stockholm. 
After his Ph.D., he was a Postdoctoral Researcher with RWTH Aachen University, Aachen, Germany, and Uppsala University, Uppsala, Sweden. 
His research interests revolve around learning and control for networked multiagent systems.
Dr.\ Baumann was the recipient of the best paper award at the 2019 ACM/IEEE International Conference on Cyber-Physical Systems, the best demo award at the 2019 ACM/IEEE International Conference on Information Processing in Sensor Systems, and the future award of the Ewald Marquardt Foundation.
He is currently co-chair of the Finland section of the joint chapter of the IEEE Control Systems, Robotics and Automation, and Systems, Man, and Cybernetics communities.
\end{IEEEbiography}
\begin{IEEEbiography}
[{\includegraphics[width=1in,height=1.25in,clip,keepaspectratio]{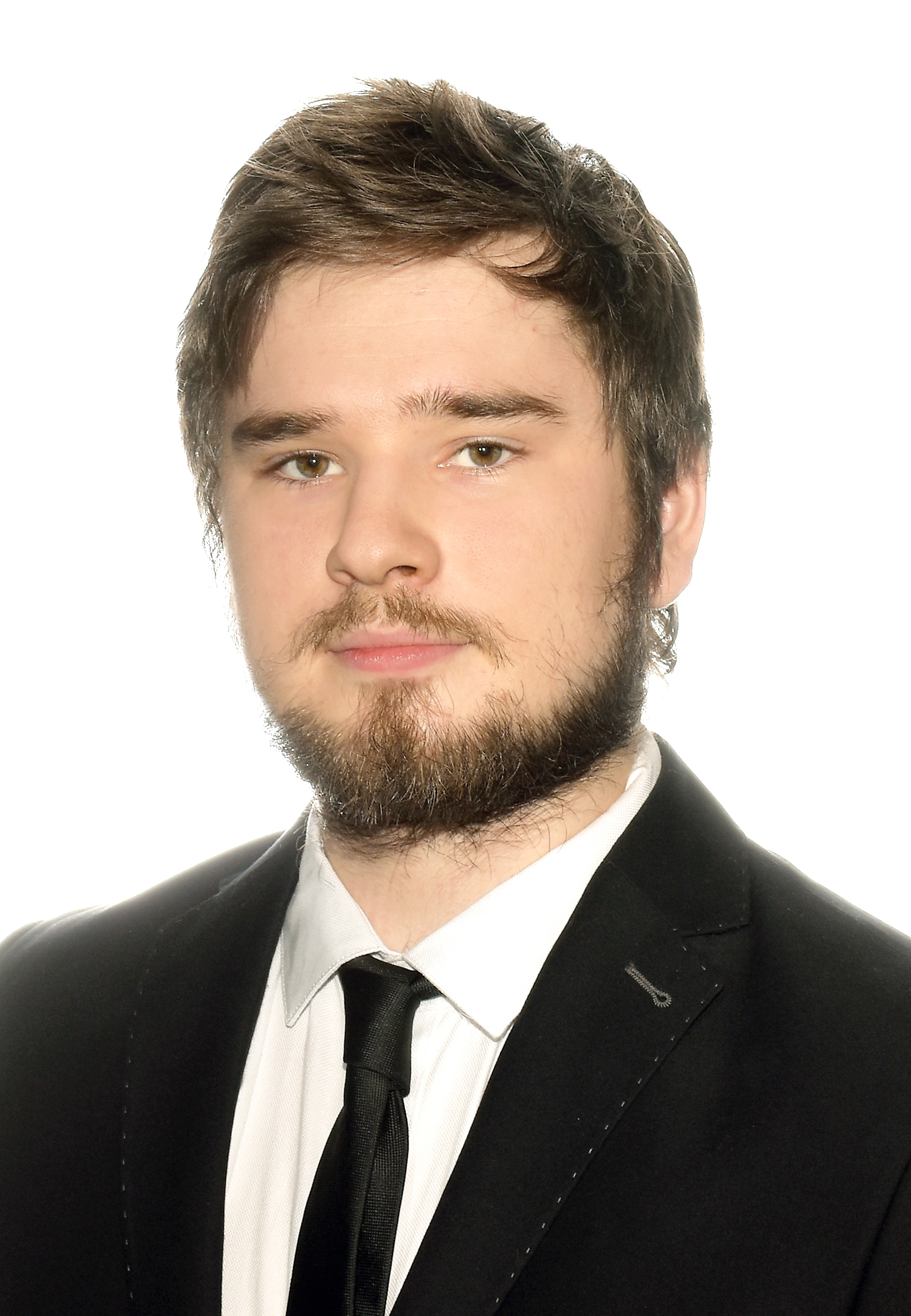}}]{Krzysztof Kowalczyk} 
received M.S. degree in control engineering and robotics in 2021 from Wroclaw University of Science and Technology, Poland. Since 2023, he has been a Ph.D candidate in the Doctoral School of Wroclaw University of Science and Technology. In his research, he focuses on nonparametric estimation techniques with high probability guarantees and their applications to multi-agent systems. 
\end{IEEEbiography}
\begin{IEEEbiography}
[{\includegraphics[width=1in,height=1.25in,clip,keepaspectratio]{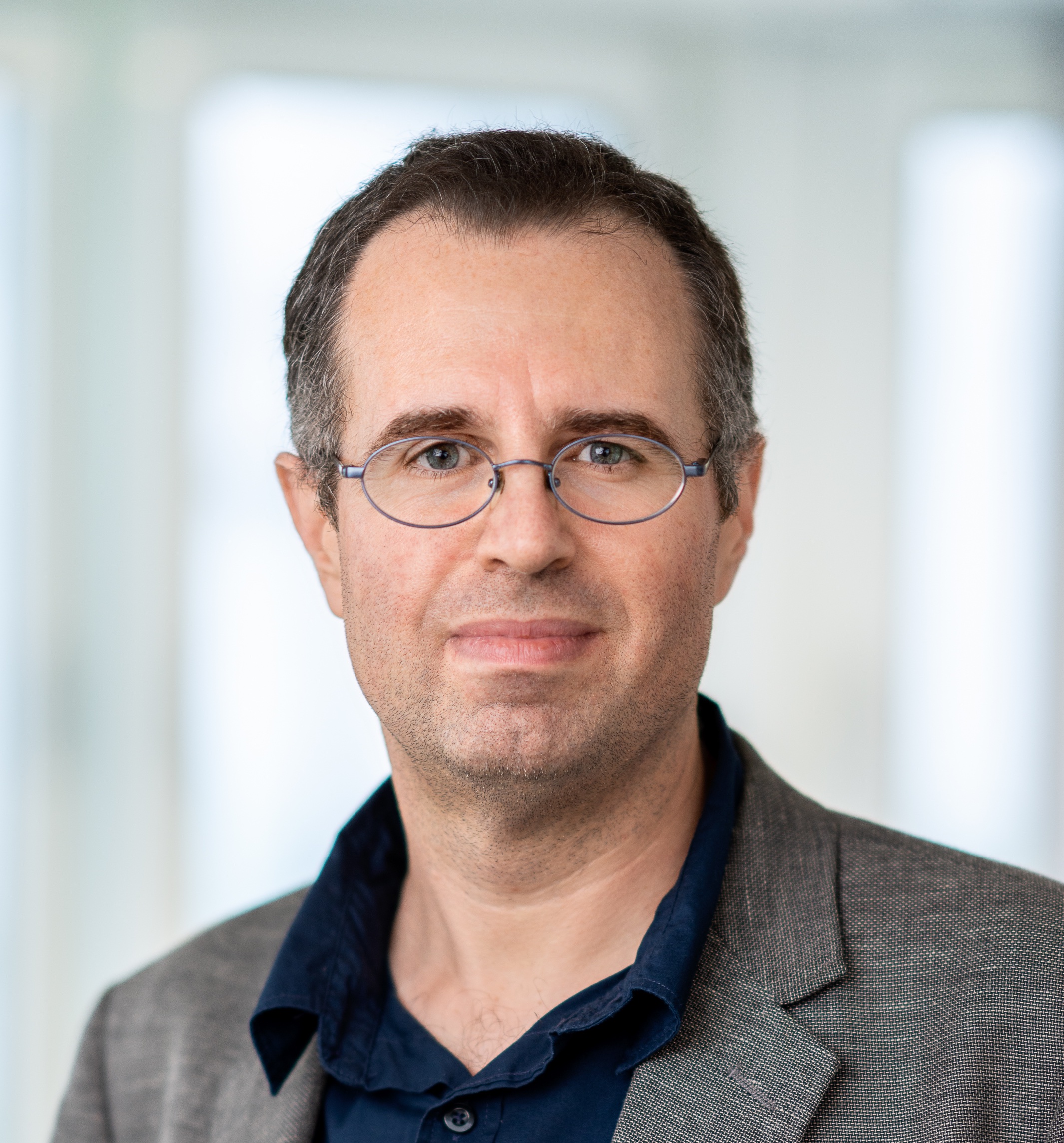}}]{Cristian R. Rojas} (Senior Member, IEEE) 
was born in 1980. He received the M.S. degree in electronics engineering from the Universidad T\'ecnica Federico Santa Mar\'\i a, Valpara\'\i so, Chile, in 2004, and the Ph.D. degree in electrical engineering at The University of Newcastle, NSW, Australia, in 2008. Since October 2008, he has been with the Royal Institute of Technology, Stockholm, Sweden, where he is currently Professor of Automatic Control at the School of Electrical Engineering and Computer Science. His research
interests lie in system identification, signal processing and machine learning. Prof. Rojas is a member of the IEEE Technical Committee on System Identification and Adaptive
Processing (since 2013), and of the IFAC Technical Committee TC1.1. on Modelling, Identification, and Signal Processing (since 2013). He has been Associate Editor for the IFAC journal Automatica and the IEEE Control Systems Letters (L-CSS).
\end{IEEEbiography}
\begin{IEEEbiography}
[{\includegraphics[width=1in,height=1.25in,clip,keepaspectratio]{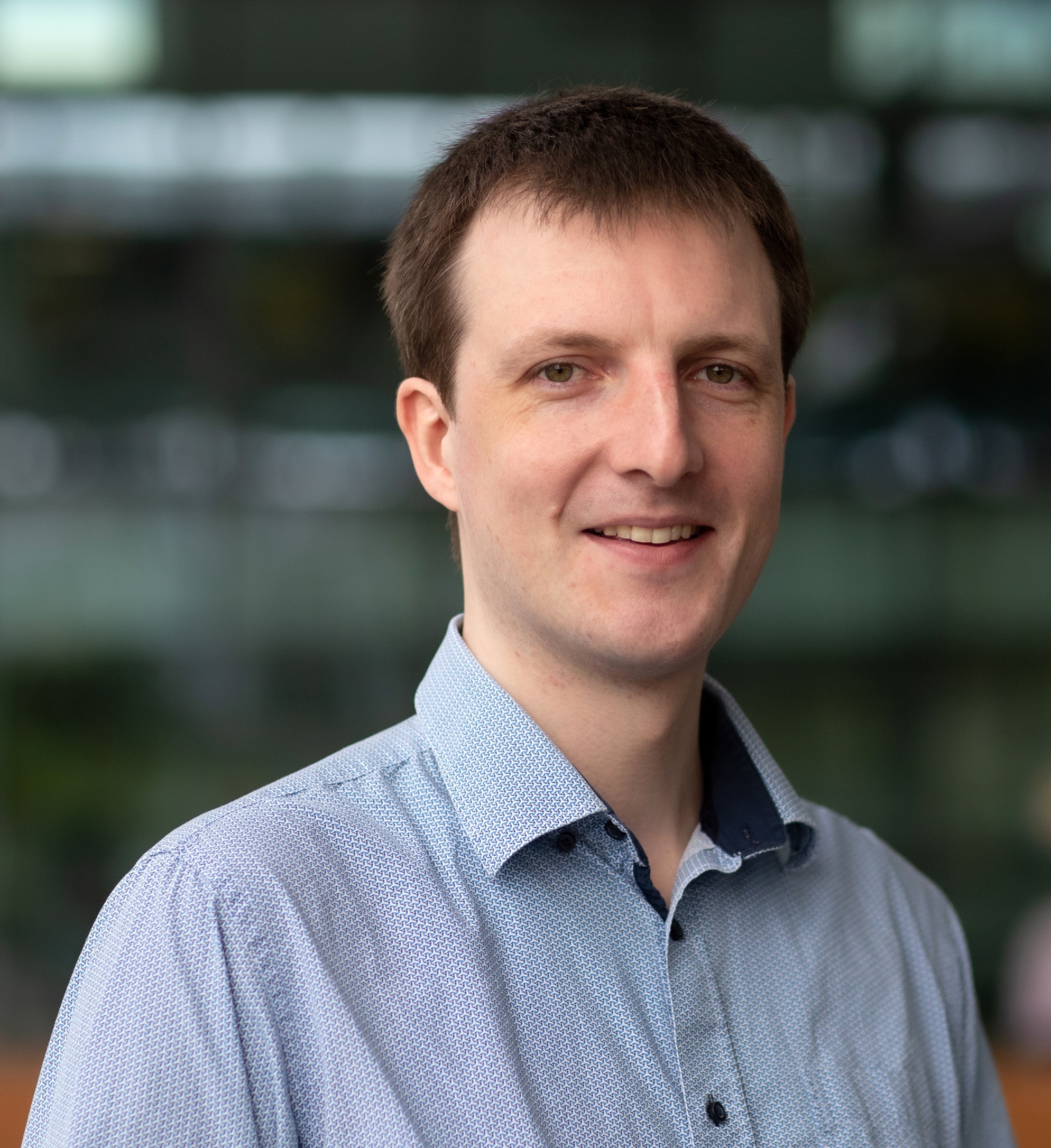}}]{Koen Tiels} (Member, IEEE) received the master’s degree in electromechanical engineering and the Ph.D. degree in engineering from the Vrije Universiteit Brussel (VUB), Brussels, Belgium, in July 2010 and March 2015, respectively. He was a Post-Doctoral Researcher with VUB from 2015 to 2018. From February 2018 to January 2020, he was with the Department of Information Technology, Division of Systems and Control, Uppsala University, Uppsala, Sweden, as a Post-Doctoral Researcher. From February 2020 to December 2022, he was an Assistant Professor with the Control Systems Technology (CST) Group, Eindhoven University of Technology (TU/e), Eindhoven, The Netherlands, where he is currently a Docent at the Department of Mechanical Engineering. His main research interests are in the field of nonlinear system identification.
\end{IEEEbiography}
\begin{IEEEbiography}
[{\includegraphics[width=1in,height=1.25in,clip,keepaspectratio]{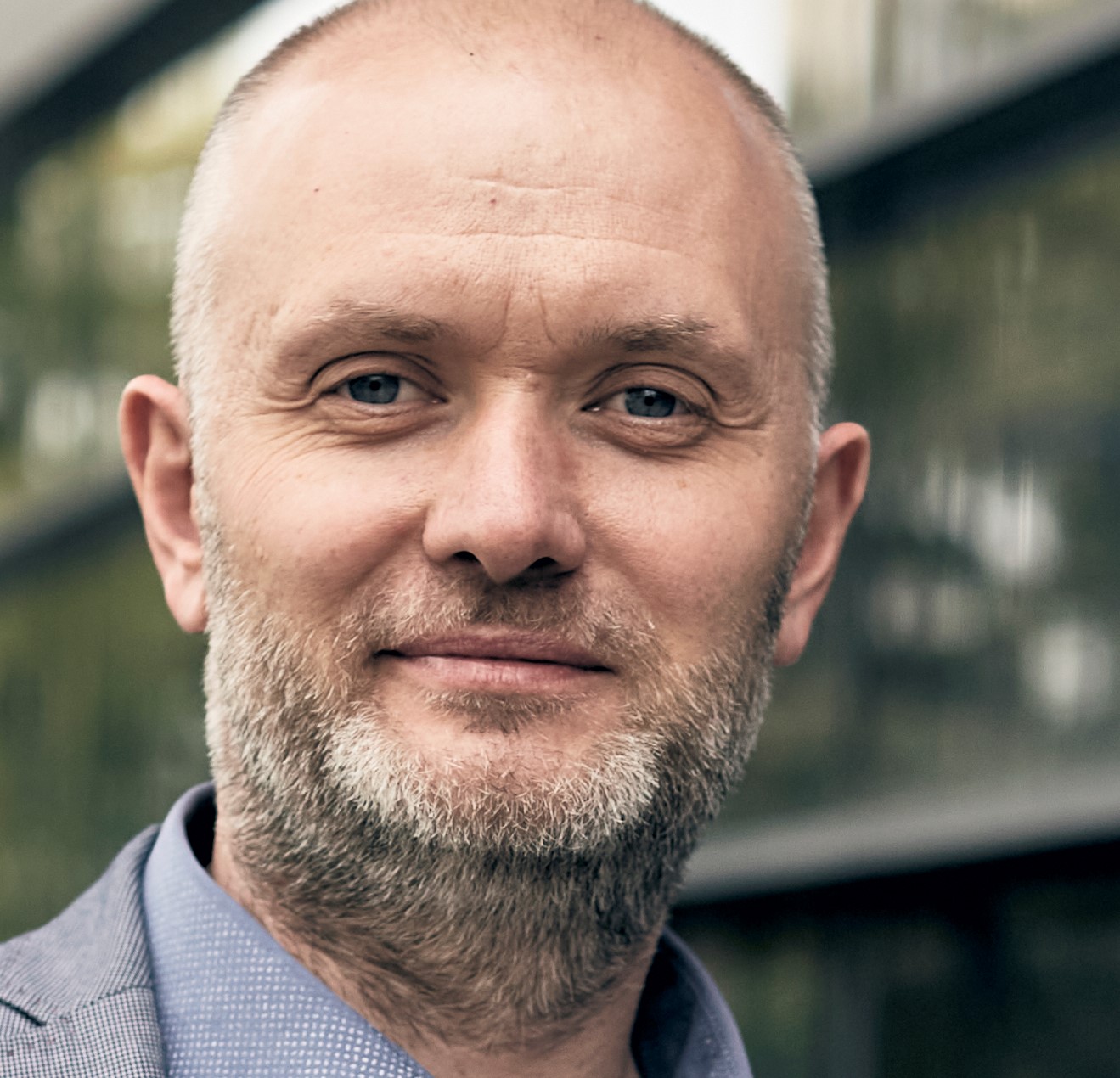}}]{Pawel Wachel} (Member, IEEE) 
     was born in 1980. He received the M.S. degree in control engineering and robotics in 2004 and the Ph.D. degree in control engineering and robotics in 2008, respectively, from Wroclaw University of Technology, Poland. 
     Since 2008, he has been with the Wroclaw University of Science and Technology, where he is currently an associate professor of computer science at the Faculty of Information and Communication Technology. 
     During his career, prof. Wachel was a postdoctoral researcher with the neuro-engineering lab (BrainLab), part of the Biomedical Signal Processing Group at the Faculty of Fundamental Problems of Technology, WUST, Poland.
     In his research, prof. Wachel focuses on system identification under limited prior knowledge, nonasymptotic aspects of machine learning, and signal processing.
\end{IEEEbiography}

\end{document}